\documentclass[conference,a4paper]{IEEEtran}
\usepackage{stfloats}
\usepackage{color}
\usepackage{amssymb} 
\usepackage{mathtools}
\usepackage{footnote}
\usepackage{amsfonts}
\newcommand{\ve}[1]{{\bf #1}}
\newcommand{\mode}{\mbox{mod }}
\newcommand{\vol}{\mbox{Vol }}
\newcommand{\norm}[1]{\left|\left|#1\right|\right|}
\newcommand{\prob}[1]{\text{Pr}\left(#1\right)}

\newtheorem{lemma}{Lemma}
\newtheorem{theorem}{Theorem}
\newtheorem{proposition}{Proposition}
\newtheorem{definition}{Definition}

\newtheorem{remark}{Remark}
\newtheorem{conjec}{Conjecture}

\ifCLASSOPTIONcompsoc \usepackage[caption=false,font=normalsize,labelfon t=sf,textfont=sf]{subfig} \else \usepackage[caption=false,font=footnotesize]{subfi
g} \fi

\begin{document}

\sloppy
\IEEEoverridecommandlockouts
\title{Gaussian Multiple Access via Compute-and-Forward}

\author{Jingge Zhu and Michael Gastpar\IEEEmembership{, Member, IEEE}
\thanks{This work was supported in part by the European ERC Starting Grant 259530-ComCom. This paper was presented in part in Zurich International Seminar on Communications, Feb. 2014 and in part in IEEE International Symposium on Information Theory, Jun. 2014.}
\thanks{J. Zhu and M. Gastpar are with the School of Computer and Communication Sciences, Ecole Polytechnique F{\'e}d{\'e}rale de Lausanne (EPFL), Lausanne,
Switzerland (e-mail: jingge.zhu@epfl.ch, michael.gastpar@epfl.ch).}

\thanks{Copyright (c) 2014 IEEE. Personal use of this material is permitted.  However, permission to use this material for any other purposes must be obtained from the IEEE by sending a request to pubs-permissions@ieee.org.}
}

\maketitle

\begin{abstract}
Lattice codes used under the Compute-and-Forward paradigm suggest an alternative strategy
for the standard Gaussian multiple-access channel (MAC): The receiver successively decodes integer
linear combinations of the messages until it can invert and recover all messages. In this paper,
a multiple-access technique called CFMA  (Compute-Forward Multiple Access) is proposed and analyzed. For the two-user MAC, it is shown that without time-sharing, the entire capacity
region can be attained using CFMA with a single-user decoder  as soon as the signal-to-noise ratios are above $1+\sqrt{2}$.   A partial analysis is given for more than two users. Lastly the strategy is
extended to the so-called dirty MAC where two interfering signals are known non-causally to the two transmitters in a distributed fashion. Our scheme extends the previously known results and gives new achievable rate regions. 
\end{abstract}

\section{Introduction}
Recent results on lattice codes applied to additive Gaussian networks show remarkable advantages of their linear structure. In particular the compute-and-forward scheme \cite{NazerGastpar_2011} demonstrates the idea that sometimes it is better to first  decode sums of several codewords than the codewords individually. Similar ideas have also been exploited in several communication networks and are shown to be beneficial in various perspectives, see for example  \cite{wilson_joint_2010} \cite{Nam_etal_2011} \cite{Nam_etal_2010} \cite{Zhan_etal_2010}.

The Gaussian multiple access channel is a well-understood communication system. To achieve its entire capacity region, the receiver can either use joint decoding (a multi-user decoder), or a single-user decoder combined with successive cancellation decoding and time-sharing  \cite[Ch. 15]{cover_elements_2006}.   An extension of the successive cancellation decoding called Rate-Splitting Multiple Access is developed in \cite{rimoldi_rate-splitting_1996} where only single-user decoders are used to achieve the whole capacity region without  time-sharing, but at the price that messages have to be split to create more virtual users.

In this paper we provide and analyze a novel strategy for the  Gaussian MAC using lattice codes. The proposed multiple-access scheme is called Compute-Forward Multiple Access (CFMA) as it is based on a modified compute-and-forward technique. For the $2$-user Gaussian MAC, the receiver first decodes the sum of the two transmitted codewords, and then decodes either one of the codewords, using the sum as side information.   As an example, Figure \ref{fig:illustration} gives an illustration of an achievable rate region for a \textit{symmetric} $2$-user Gaussian MAC with our proposed scheme.  When the \textit{signal-to-noise ratio} (SNR) of both users is below $1.5$, the proposed scheme cannot attain rate pairs on the dominant face of the capacity region. If the SNR exceeds $1.5$, a line segment on the capacity boundary is achievable. As SNR increases, the end points of the line segment approach the corner points, and the whole capacity region is achievable as soon as the SNR of both users is larger than $1+\sqrt{2}$. We point out that the decoder used in our scheme is a single-user decoder since it mainly performs lattice quantizations on the received signal, in contrast to joint decoding where the decoder  needs the complete information of the codebooks of the two users.  Hence this novel approach allows us to achieve rate pairs in the capacity region using only  single-user decoders (with interference cancellation), while time-sharing or rate splitting are not needed. This feature of the proposed coding scheme could be of interest for practical considerations.

\begin{figure}[htb!]
\centering
\includegraphics[width=2.5in]{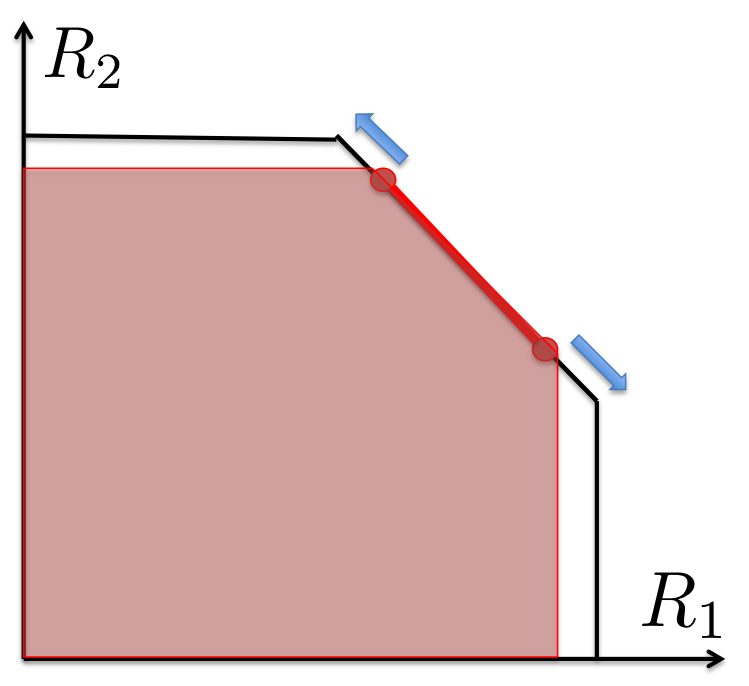}
\caption{An illustration of an achievable rate region for a symmetric $2$-user Gaussian MAC with the proposed CFMA scheme. The rate pairs in the shaded region can be achieved using a single-user decoder without time-sharing. As SNR increases, the end points of the line segment approach the corner points and the whole capacity region becomes achievable. A sufficient condition for achieving the whole capacity region is that  the SNR of both users are above $1+\sqrt{2}$ (in the symmetric case).}
\label{fig:illustration}
\end{figure}

We should point out that a related result in \cite{ordentlich_approximate_2014} shows that using a similar idea of decoding multiple integer sums, the sum capacity of the Gaussian MAC can be achieved within a constant gap. Furthermore, it is also shown in \cite{ordentlich_successive_2013}  that under certain conditions, some isolated (non-corner) points of the capacity region can be attained. To prove these results, the authors use fixed lattices which are independent of channel gains.  In this work, we close these gaps by showing that if the lattices are properly scaled in accordance with the channel gains, the full capacity region can be attained. Moreover, this paper considers exclusively the Gaussian MAC. Related results for the general discrete-memoryless MAC are given in \cite{Nazer_compute_discrete_2014} \cite{Lim_etal_2016} using a joint typicality approach.

The proposed coding scheme is also extended to the general $K$-user Gaussian MAC and achievable rate regions are derived. While a complete characterization of the achievable region is difficult to give for the general case, we raise a conjecture that the symmetric capacity\footnote{For a symmetric $K$-user Gaussian MAC where the SNR of all users equals $P$, we say that the symmetric capacity is achievable if each user has a rate $\frac{1}{2K}\log(1+KP)$.} is always achievable for the symmetric Gaussian MAC, provided the SNR exceeds a certain threshold.  While this conjecture is  established for the $2$-user case where the SNR threshold is $1.5$,  some numerical evidence is  given to support this conjecture for larger $K$. For example the numerical results show that the SNR threshold is less than $2.24$  for the $3$-user symmetric MAC and less than $3.75$ for the $4$-user symmetric MAC.

We then study the so-called ``dirty'' Gaussian MAC with two additive interference signals which are non-causally known to two encoders in a distributed manner. It was shown in \cite{philosof_lattice_2011} that lattice codes are well-suited for this problem. We devise a coding scheme within our framework to this system which extends the previous results and gives a new achievable rate region, which could be considerably larger for general interference strength.  

Lastly we should point out that although this paper only considers the  multiple access channel, the proposed coding scheme is more general and can be applied to many other Gaussian network problems.

The paper is organized as follows. Section \ref{sec:system} gives the problem statement and introduces the nested lattice codes used in our coding scheme. The important notion of computation rate tuple is also introduced. Section \ref{sec:MAC_2usr} gives a complete analysis of our coding scheme in the $2$-user Gaussian MAC. In Section \ref{sec:K_user} we extend the coding scheme to the $K$-user Gaussian MAC. A similar strategy is then applied to the Gaussian dirty MAC in Section \ref{sec:dirty} for the two user case.

Throughout this paper, vectors and matrices are denoted by  lowercase and uppercase bold letters, such as $\ve a$ and $\ve A$, respectively. The $(i,j)$-entry of a matrix $\ve A$ is denoted by $A_{ij}$. The notation $\text{diag}(x_1,\ldots,x_K)$ denotes a diagonal matrix whose diagonal entries are $x_1,\ldots,x_K$.  The determinant of a matrix $\ve A$ is denoted by $|\ve A|$. The probability of a given event $E$ is denoted by $\mathbb P(E)$.

\section{Nested lattice codes and computation rate tuples}\label{sec:system}
We  consider a $K$-user Gaussian multiple access channel. The discrete-time real Gaussian MAC has the following vector representation
\begin{IEEEeqnarray}{rCl}
\ve y=\sum_{k=1}^K h_k\ve x_k+\ve z
\label{eq:sys_K_MAC}
\end{IEEEeqnarray}
with $\ve y, \ve x_k\in\mathbb R^n$ denoting the channel output at the receiver and channel input of transmitter $k$. The white Gaussian  noise with unit variance per entry is denoted by $\ve z\in\mathbb R^n$. A fixed real number $h_k$ denotes the channel coefficient from user $k$ to the receiver and is known to transmitter $k$.  We can assume without loss of generality that every user has the same power constraints on the channel input as $\mathbb E\{\norm{\ve x_k}^2\}\leq nP$. 

In our coding scheme,  we map messages $W_k$ of user $k$ bijectively to  points in $\mathbb R^n$ denoted by $\ve t_k$, which are elements of the codebook $\mathcal C_k$ to be defined later. The \textit{rate} of the codebook $\mathcal C_k$ is defined to be
\begin{align}
r_k:=\frac{1}{n}\log|\mathcal C_k|\quad\text{ for }k=1, 2
 \label{eq:message_rate}
\end{align}
Each transmitter is equipped with an encoder $\mathcal E_k$ which maps its message (or the corresponding codeword) to the channel input as $\ve x_k=\mathcal E_k(\ve t_k)$.  At the receiver, a decoder  wishes to estimate all the messages using the channel output $\ve y$. The decoded codewords are denoted by $\hat{\ve t}_1,\ldots,\hat{\ve t}_K$ and they are mapped back to messages.  Hence we can define the message error probability as
\begin{align}
P_{e,msg}^{(n)}:= \mathbb P\left(\bigcup_{k=1}^K\{\hat{\ve t}_k\neq \ve t_k\}\right)
\label{eq:error_message}
\end{align}
where $n$ is the length of  codewords. We require the receiver  to decode all messages from $\ve y$ with an arbitrarily small error. Formally we have the following definition.
\begin{definition}[Message rate tuple]
Consider a $K$-user Gaussian MAC in (\ref{eq:sys_K_MAC}). We say \textit{a message rate tuple $(R_1,\ldots, R_K)$}  is achievable if it holds that
for any $\epsilon>0$, there exists a number $n_0$ such that for all $n\geq n_0$,   the message error probability  in (\ref{eq:error_message}) satisfies $P_{e,msg}^{(n)}<\epsilon$ whenever the rate  $r_k$  defined  in (\ref{eq:message_rate}) of user $k$ satisfies $r_k< R_k$ for $k=1,\ldots, K$.
\label{def:message_rate}
\end{definition}

The capacity region, equivalently all possible achievable message rate tuples, of a $K$-user Gaussian MAC is known, see for example \cite{ahlswede_multi-way_1973}, \cite{liao_multiple_1972}, \cite[Ch. 15]{cover_elements_2006}. In this paper we devise a novel approach to achieve the capacity region of the Gaussian MAC.

\subsection{Nested lattice codes}
In this section we describe the encoding procedure of our scheme based on nested lattice codes. We state the main facts about nested lattice codes here, and more details can be found in \cite{ErezZamir_2004} \cite{Erez_etal_2005}.

A lattice $\Lambda$ is a discrete subgroup of $\mathbb R^n$ with the property that if $\ve t_1, \ve t_2\in \Lambda$, then $\ve t_1+\ve t_2\in \Lambda$. Define the lattice quantizer $Q_{\Lambda}: \mathbb R^n\rightarrow\Lambda$ as
\begin{IEEEeqnarray*}{rCl}
Q_{\Lambda}(\ve x)=\mbox{argmin}_{\ve t\in\Lambda}\norm{\ve t-\ve x}
\end{IEEEeqnarray*}
and define the fundamental Voronoi region of the lattice to be
\begin{IEEEeqnarray*}{rCl}
\mathcal V:=\{\ve x\in\mathbb R^n:Q_{\Lambda}(\ve x)=\ve 0\}
\end{IEEEeqnarray*}
The modulo operation gives the quantization error:
\begin{IEEEeqnarray*}{rCl}
[\ve x]\mode\Lambda=\ve x-Q_{\Lambda}(\ve x)
\label{eq:mode}
\end{IEEEeqnarray*}
Two lattices $\Lambda$ and $\Lambda'$ are said to be nested if $\Lambda'\subseteq\Lambda$. 

Let $\beta_k, k=1,\ldots, K$ be $K$ nonzero real numbers and we collect them into one vector $\underline{\beta}:=(\beta_1,\ldots,\beta_K)$. In a general $K$-user Gaussian MAC, for each user we choose a lattice $\Lambda_k$ which is \textit{good for AWGN channel coding} in the sense of \cite{Erez_etal_2005}. These $K$ lattices $\Lambda_k, k=1,\ldots, K$ can be chosen to form a nested lattice chain \cite{Nam_etal_2011} according to certain order to be determined later. We let $\Lambda_c$ denote the coarsest lattice among them, i.e., $\Lambda_c\subseteq\Lambda_k$ for all $k=1,\ldots,K$. We can also construct $K$ lattices $\Lambda_k^{s}\subseteq\Lambda_c$ for all $k$ where all lattices  are \textit{simultaneously good}  in the sense of \cite{Erez_etal_2005},  and with second moment 
\begin{align*}
\frac{1}{n\vol(\mathcal V_k^s)}\int_{\mathcal V_k^{s}}\norm{\ve x}^2d\ve x=\beta_k^2 P
\end{align*}
where $\mathcal V_k^{s}$ denotes the Voronoi region of the lattice $\Lambda_k^{s}$.   The lattice $\Lambda_k^s$ is used as the \textit{shaping region} for the codebook of user $k$.

For each transmitter $k$, we construct the codebook as
\begin{IEEEeqnarray}{rCl}
\mathcal C_k=\Lambda_k\cap\mathcal V_k^s
\label{eq:codebook}
\end{IEEEeqnarray}
With this codebook the message rate of user $k$ is
\begin{IEEEeqnarray}{rCl}
r_k=\frac{1}{n}\log|\mathcal C_k|=\frac{1}{n}\log\frac{\vol(\mathcal V_k^s)}{\vol(\mathcal V_k)}
\label{eq:message_rate_lattice}
\end{IEEEeqnarray}
where $\mathcal V_k$ is the Voronoi region of the fine lattice $\Lambda_k$.

The parameters $\underline{\beta}$  are used to control the individual rates of different users. We will see later that the proper choice of these parameters depend on the channel coefficients. We also note that a similar idea appears in \cite{ntranos_asymmetric_2013} \cite{ordentlich_precoded_2013} whereas the authors do not make connections between these parameters  and  channel coefficients. 

\subsection{The computation rate tuple}

Throughout this work, we will be interested in decoding functions of codewords. One important example is the sum of the lattice codewords of the form
\begin{align}
\ve u:=\left[\sum_{k=1}^Ka_k\ve t_k\right]\mod \Lambda_f^s
\label{eq:modulo_sum}
\end{align}
where $\Lambda_f^s$ denotes the finest lattice among $\Lambda_k^s$ and $a_k$ is an integer, $k=1,\ldots, K$.  Let $\hat{\ve u}$ denote the decoded integer sum at the receiver and define the error probability of decoding a sum as
\begin{align}
P_{e,sum}^{(n)}:=\mathbb P(\hat{\ve  u}\neq \ve u)
\label{eq:error_sum}
\end{align}
where $n$ is the length of codewords. This idea is, in the first place, different from the usual decoding procedure where individual messages are decoded. To articulate the point, we give a definition of the computation rate tuple in the context of the $K$-user Gaussian MAC.  
\begin{definition}[Computation rate tuple]
Consider a $K$-user Gaussian MAC in (\ref{eq:sys_K_MAC}). We say \textit{a computation rate tuple $(R_1^{\ve a},\ldots, R_K^{\ve a})$ with respect to the sum (\ref{eq:modulo_sum})}  is achievable if it holds for any $\epsilon>0$, there exists a number $n_0$  such that for all $n\geq n_0$, the sum decoding error probability in (\ref{eq:error_sum}) satisfies $P_{e,sum}^{(n)}<\epsilon$ whenever the rate  $r_k$  defined  in (\ref{eq:message_rate}) of user $k$ satisfies $r_k< R_k^{\ve a}$ for $k=1,\ldots, K$.
\label{def:computation_rate}
\end{definition}

An achievable computation rate tuple in the Gaussian MAC is given in the following theorem, as a generalization of the result of \cite{NazerGastpar_2011}. 
\begin{theorem}[A general compute-and-forward formula]
Consider a $K$-user Gaussian MAC with channel coefficients $\ve h=(h_1,\ldots, h_K)$ and equal power constraint $P$.  Let  $\beta_1,\ldots,\beta_K$ be $K$ nonzero real numbers. The computation rate tuple $(R_1^{\ve a},\ldots, R_K^{\ve a})$ with respect to the sum (\ref{eq:modulo_sum})  is achievable with
\begin{IEEEeqnarray}{rCl}
R_k^{\ve a}= \left[\frac{1}{2}\log \left(\norm{\ve{\tilde a}}^2-\frac{P(\ve h^T\ve{\tilde a})^2}{1+P\norm{\ve h}^2}\right)^{-1}+\frac{1}{2}\log \beta_k^2\right]^+
\label{eq:compute_rate}
\end{IEEEeqnarray}
where $\ve{\tilde a}:=[\beta_1a_{1},...,\beta_Ka_{K}]$ and $a_{k}\in\mathbb Z$  for all $k\in[1:K]$.
\label{thm:computation_rate}
\end{theorem}
\begin{proof}
A  proof is given in Appendix \ref{sec:appen_proof_computation}.
\end{proof}

\begin{remark}
\
\begin{itemize}
\item By setting $\beta_k=1$ for all $k$ we recover the original compute-and-forward formula given in \cite{NazerGastpar_2011} Theorem 4.
\item The usefulness of the parameters $\beta_1,\ldots,\beta_K$ lies in the fact that they can be chosen according to the channel coefficients $h_k$ and power $P$. This is crucial to our coding scheme for a Gaussian MAC.
\item This formula also illustrates why it is without loss of generality to assume that all powers are equal. In the case that  each transmitter has power $P_k$,  just replace $h_k$ by $h_k':=\sqrt{P_k/P}h_k$ for all $k$ in (\ref{eq:compute_rate}).
\item It is straightforward to extend the result when there are multiple receivers, see \cite{zhu_thesis} \cite{ZhuGastpar_2014}.
\end{itemize}
\end{remark}

Before moving on, it is instructive to inspect  formula (\ref{eq:compute_rate}) in some detail. This will give some insights on why this modified scheme will be helpful for a multiple-access channel. To do this, we can rewrite (\ref{eq:compute_rate}) in the following expression
\begin{align*}
R_k^{\ve a}= &\frac{1}{2}\log\left(\beta_i(1+P\|\mathbf h\mathbf\|^2)\right) \\
&-\frac{1}{2}\log\left(\|\mathbf{\tilde a}\|^2+ P(\|\mathbf h\|^2\|\mathbf{\tilde a}\|^2-(\mathbf h^T \mathbf{\tilde a})^2)\right).
\end{align*} 
As already pointed out in \cite{NazerGastpar_2011}, the term $\|\mathbf h\|^2\|\mathbf{\tilde a}\|^2-(\mathbf h^T \mathbf{\tilde a})^2$ in the second log  has a natural interpretation -- it measures how the  coefficients $\tilde{\ve a}$ differs from the channel $\ve h$, in other words the rate loss occurred because of the mismatch between the chosen sum coefficients and  channel gains. Cauchy-Schwartz inequality implies that this term is always nonnegative and is zero if and only if $\tilde{\ve a}$ is colinear with the channel coefficients $\ve h$.  Notice that in the original compute-and-forward scheme, where $\tilde{\ve a}=\ve a$ by setting all $\beta_k$ to be $1$, this term is not necessarily zero because $\ve a$ is an integer vector while $\ve h$ can take all possible values in $\mathbb R^K$. However in this generalized scheme we are given the freedom to tune parameters $\beta_k\in\mathbb R^K$, and the rate loss due to the mismatch can be completely eliminated by choosing $\beta_k$ to align $\tilde {\ve a}$ with $\ve h$. In general, the lattice scaling coefficients $\beta_k$ allow us to adjust the codebook rate freely and is essential to our coding scheme for the Gaussian MAC discussed in the sequel.

\subsection{Message rate tuple vs. computation rate tuple}
It is important to distinguish the achievable \textit{message rate tuple}  in Definition \ref{def:message_rate}, where individual messages should be decoded, and the achievable  \textit{computation rate tuple} in Definition \ref{def:computation_rate}, where 
 only one function of messages is decoded. The superscript $\ve a$ in the notation $R_k^{\ve a}$ is used to emphasize the different decoding goals. We give an example of computation rate pairs for a $2$-user Gaussian MAC in Figure \ref{fig:CompRateRegion}. It is worth noting that the achievable computation rate region can be strictly larger than the achievable message rate region.
 
\begin{figure}[h!]
\centering
\includegraphics[scale=0.45]{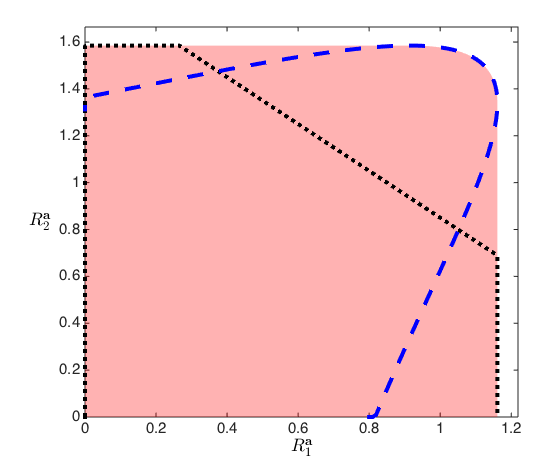}
\caption{In this figure we show an achievable computation rate region for computing the sum $[\ve t_1+\ve t_2]\mod\Lambda_f^s$ over a $2$-user Gaussian MAC where $h_1=1, h_2=\sqrt{2}$ and $P=4$. The dotted black line shows the capacity region of this MAC. The dashed blue line depicts the evaluated computation rate pairs given in (\ref{eq:compute_rate}) in Theorem \ref{thm:computation_rate} where points along this curve are obtained by choosing different $\beta_1, \beta_2$.  The shaded region shows the whole computation rate region, in which all the computation rate pairs are achievable. Notice that in this case the computation rate region contains the whole capacity region and is strictly larger than the latter. We also point out that if a computation rate pair $(R_1^{\ve a},R_2^{\ve a})$ is achievable, then any computation rate pair $(R_1', R_2')$ with $R_1'\leq R_1^{(\ve a)}, R_2'\leq R_2^{(\ve a)}$ is also achievable. This observation is used to compute the shaded region in the plot based on the curve of computation rate pairs.}
\label{fig:CompRateRegion}
\end{figure}

\section{The 2-user Gaussian MAC}\label{sec:MAC_2usr}

In this section we study the 2-user Gaussian MAC
\begin{IEEEeqnarray}{rCl}
\ve y=h_1\ve x_1+h_2\ve x_2+\ve z
\label{eq:system_clean_MAC}
\end{IEEEeqnarray}
with other specifications  given in (\ref{eq:sys_K_MAC}). We will give a complete characterization of the achievable rate region under our coding scheme.

\begin{itemize}
\item Encoding: For user $k$, given the message and the unique corresponding codeword $\ve t_k$, the channel input is generated as
\begin{align}
\ve x_k=[\ve t_k/\beta_k+\ve d_k]\mode\Lambda_k^s/\beta_k, k=1,2.
\end{align}
where $\ve d_k$ is called a \textit{dither} which is a random vector uniformly distributed in the scaled Voronoi region $\mathcal V_k^s/\beta_k$.
\item Decoding: To decode the first sum with coefficients $(a_1, a_2)$, let $\Lambda_f$ denote the finer lattice between $\Lambda_1, \Lambda_2$ if $a_1,a_2\neq 0$. Otherwise set $\Lambda_f=\Lambda_1$ if $a_2=0$, or $\Lambda_f=\Lambda_2$ if $a_1=0$. Let $\alpha_1$ be a real number to be determined later and form $\tilde {\ve y}_1:=\alpha_1\ve y-\sum_ka_{k}\beta_k\ve d_k$, the first sum with coefficient $\ve a$ is decoded by performing the lattice quantization
\begin{align*}
Q_{\Lambda_f}(\tilde {\ve y}_1).
\end{align*}
Define $\Lambda_f'$ in the similarly way for the second sum with coefficients $(b_1,b_2)$, the second sum is obtained by performing the lattice quantization
\begin{align*}
Q_{\Lambda_f'}(\tilde{\ve y}_2)
\end{align*} 
where the construction of $\tilde{\ve y}_2$ is given the proof of the following theorem. 
\end{itemize}

\begin{theorem}[Achievable message rate pairs]
Consider the $2$-user multiple access channel in (\ref{eq:system_clean_MAC}). The following message rate pair is achievable
\begin{align*}
R_k= 
\begin{cases}
r_k(\ve a, \underline{\beta}) &\quad\mbox{if }b_k=0\\
r_k(\ve b|\ve a, \underline{\beta}) &\quad\mbox{if }a_k=0\\
\min\{r_k(\ve a, \underline{\beta}), r_k(\ve b|\ve a, \underline{\beta})\}&\quad\mbox{otherwise}
\end{cases}
\end{align*}
for any linearly independent $\ve a,\ve b\in\mathbb Z^2$ and $\underline{\beta}\in\mathbb R^2$ if it holds $r_k(\ve a, \underline{\beta})\geq 0$ and $r_k(\ve b|\ve a, \underline{\beta})\geq 0$ for $k=1,2$, where we define
\begin{IEEEeqnarray}{rCl}
r_k(\ve a,\underline{\beta})&:=&\frac{1}{2}\log\frac{\beta_k^2(1+h_1^2P+h_2^2P)}{K(\ve a,\underline{\beta})}\label{eq:rate_a_clean}\\
r_k(\ve b|\ve a,\underline{\beta})&:=&\frac{1}{2}\log\frac{\beta_k^2 K(\ve a,\underline{\beta})}{\beta_1^2\beta_2^2(a_2b_1-a_1b_2)^2}\label{eq:rate_b_clean}
\end{IEEEeqnarray}
with
\begin{align}
K(\ve a,\underline{\beta}):=\sum_{k=1}^2a_k^2\beta_k^2+P(a_1\beta_1h_2-a_2\beta_2h_1)^2
\label{eq:contant_K}
\end{align}
\label{thm:MAC_general}
\end{theorem}

\begin{proof}
Recall that the transmitted signal for user $k$ is given by
\begin{IEEEeqnarray}{rCl}
\ve x_k=[\ve t_k/\beta_k+\ve d_k]\mode\Lambda_k^s/\beta_k
\end{IEEEeqnarray}
 As pointed out in \cite{ErezZamir_2004}, $\ve x_k$ is independent of $\ve t_k$ and  uniformly distributed in $\Lambda_k^s/\beta_k$ hence  has average power $P_k$ for $k=1,2$.

Given two integers $a_1, a_2$ and some real number $\alpha_1$,  we can form
\begin{IEEEeqnarray*}{rCl}
\tilde {\ve y}_1&:=&\alpha_1\ve y-\sum_ka_{k}\beta_k\ve d_k\\
&=&\underbrace{\sum_k(\alpha_1 h_{k}-a_{k}\beta_k)\ve x_k+\alpha_1\ve z_1}_{\tilde{\ve z}_1}+\sum_k a_{k}\beta_k\ve x_k-\sum_ka_{k}\beta_k\ve d_k\\
&\stackrel{(a)}{=}&\tilde{\ve  z}_1+\sum_k a_{k}\left(\beta_k (\ve t_k/\beta_k+\ve d_k)-\beta_k Q_{\Lambda_k^s/\beta_k}(\ve t_k/\beta_k+\ve d_k)\right)\\
&&-\sum_ka_{k}\beta_k\ve d_k\\
&\stackrel{(b)}{=}&\tilde {\ve z}_1+\sum_k a_{k} (\ve t_k-Q_{\Lambda_k^s}(\ve t_k+\beta_k\ve d_k))\\
&=&\tilde {\ve z}_1+\sum_k a_{k} \tilde{\ve t}_k \label{eq:y(1)}
\end{IEEEeqnarray*}
with the notation
\begin{IEEEeqnarray}{rCl}
\tilde{\ve z}_1&:=&\sum_k(\alpha_1h_k-\beta_ka_k)\ve x_k+\alpha_1\ve z\\
\tilde{\ve t}_k&:=&\ve t_k-Q_{\Lambda_k^s}(\ve t_k+\beta_k\ve d_k)
\label{eq:t_tilde}
\end{IEEEeqnarray}
Step (a) follows from the definition of $\ve x_k$ and step (b) uses the identity $Q_{\Lambda}(\beta\ve x)=\beta Q_{\Lambda/\beta}(\ve x)$ for any real number $\beta\neq 0$ and any lattice $\Lambda$. Note that $\tilde{\ve t}_k$ lies in $\Lambda_f$ due to the nested construction $\Lambda_k^s\subseteq\Lambda_f$. The term $\tilde{\ve z}_1$ acts as an equivalent noise independent of $\sum_k a_k\tilde{\ve t}_k$ (thanks to the dithers) and has an average variance per dimension
\begin{IEEEeqnarray}{rCl}
N_1(\alpha_1)=\sum_k(\alpha_1h_1-\beta_ka_k)^2P+\alpha_1^2.
\end{IEEEeqnarray}
The decoder obtains the sum $\sum_ka_k\tilde{\ve t}_k$ from $\tilde{\ve y}_1$ using \textit{lattice decoding}: it quantizes $\tilde{\ve y}_1$ to its closest lattice point in $\Lambda_f$. Using the same  argument in the proof of Theorem \ref{thm:computation_rate} we can show this decoding process is successful if the rate of the transmitter $k$ satisfies
\begin{IEEEeqnarray}{rCl}
r_k&< &r_k(\ve a, \underline{\beta})=\max_{\alpha_1}\frac{1}{2}\log^+\frac{\beta_k^2P}{N_1(\alpha_1)}
\label{eq:rate_a_clean_proof}
\end{IEEEeqnarray}
Optimizing over $\alpha_1$ we obtain the claimed expression in (\ref{eq:rate_a_clean}). In other words we have the computation rate pair $(R_1^{\ve a}:=r_1(\ve a,\underline{\beta}), R_2^{\ve a}:=r_2(\ve a,\underline{\beta}))$. \footnote{Strictly speaking, the computation rate pair is defined under the condition that the sum $[\sum_ka_k\ve t_k]\mod\Lambda_f^s$ can be decoded in Definition \ref{def:computation_rate}.  Here we actually decode the sum  $\sum_ka_k\tilde{\ve t}_k$.  However this will not affect the achievable message rate pair, because we can also  recover the two codewords $\ve t_1$ and $\ve t_2$ using the two sums   $\sum_ka_k\tilde{\ve t}_k$ and  $\sum_kb_k\tilde{\ve t}_k$, as shown in the proof. } We remark that the expression  (\ref{eq:rate_a_clean}) is exactly the general compute-and-forward formula given in Theorem \ref{thm:computation_rate} for $K=2$.

To decode a second integer sum with coefficients $\ve b$ we use the idea of successive cancellation \cite{NazerGastpar_2011}\cite{Nazer_2012}. If $r_k(\ve a,\underline{\beta})>0$ for $k=1,2$, i.e., the sum $\sum_ka_k\tilde{\ve t}_k$ can be decoded, we can reconstruct the term $\sum_ka_k\beta_k\ve x_k$ as $\sum_ka_k\beta_k\ve x_k=\sum_ka_k\tilde{\ve t}_k+\sum_k a_k\beta_k\ve d_k$. Similar to the derivation of (\ref{eq:y(1)}), we can use $\sum_ka_k\beta_k\ve x_k$  to form
\begin{IEEEeqnarray}{rCl}
\tilde{\ve y}_2&:=&\alpha_2\ve y+\lambda(\sum_ka_k\beta_k\ve x_k)-\sum_kb_k\beta_k\ve d_k\\
&=&\sum_k(\alpha_2h_k-(b_k+\lambda a_k)\beta_k)\ve x_k+\alpha_2\ve z+\sum_kb_k\tilde{\ve t}_k\\
&=&\tilde{\ve z}_2+\sum_kb_k\tilde{\ve t}_k
\end{IEEEeqnarray}
where the equivalent noise 
\begin{align}
\tilde{\ve z}_2:=\sum_k(\alpha_2h_k-(b_k+\lambda a_k)\beta_k)\ve x_k+\alpha_2\ve z
\end{align}
has average power per dimension
\begin{align}
N_2(\alpha_2,\lambda)=\sum_k(\alpha_2h_k-(b_k+\lambda a_k)\beta_k)^2P+\alpha_2^2.
\end{align}
Under lattice decoding with respect to $\Lambda_f'$, the term $\sum_k b_k\tilde{\ve t}_k$ can be decoded if for $k=1, 2$ we have
\begin{IEEEeqnarray}{rCl}
r_k&< &r_k(\ve b|\ve a,\underline{\beta})=\max_{\alpha_2,\lambda}\frac{1}{2}\log^+\frac{\beta_k^2P}{N_2(\alpha_2,\lambda)}
\label{eq:rate_b_clean_proof}
\end{IEEEeqnarray}
Optimizing over $\alpha_2$ and $\lambda$ gives the claimed expression in (\ref{eq:rate_b_clean}).  In other words we have the computation rate pair $(R_1^{\ve b}:=r_1(\ve b|\ve a,\underline{\beta}), R_2^{\ve b}:=r_2(\ve b|\ve a,\underline{\beta}))$.

A simple yet important observation is that if $\ve a, \ve b$ are two linearly independent vectors, then $\tilde{\ve t}_1$ and $\tilde{\ve t}_2$ can be solved using the two decoded sums, and consequently two codewords $\ve t_1,\ve t_2$ are found as $\ve t_k=[\tilde{\ve t}_k]\mod \Lambda_k^s$. This means that if two vectors $\ve a $ and $\ve b$ are linearly independent, the message rate pair $(R_1,R_2)$ is achievable with
\begin{align}
R_k=\min\{r_k(\ve a, \underline{\beta}), r_k(\ve b|\ve a, \underline{\beta})\}
\label{eq:message_rate_proof}
\end{align}
Another important observation is that when we decode a sum $\sum_ka_k\tilde{\ve t}_k$ with the coefficient $a_i=0$,  the lattice point $\tilde{\ve t}_i$ does not participate in the sum $\sum_ka_k\tilde{\ve t}_k$ hence the rate $R_i$ will not be constrained by this decoding procedure as in (\ref{eq:rate_a_clean_proof}). For example if we decode $a_1\tilde{\ve t}_1+a_2\tilde{\ve t}_2$ with $a_1=0$, the computation rate pair is actually $(R_1^{\ve a}=\infty, R_2^{\ve a}=r_1(\ve a,\underline{\beta}))$, since the rate of user $1$ in this case can be arbitrarily large. The same argument holds for the case  $b_k=0$. Combining (\ref{eq:message_rate_proof}) and the special cases when $a_k$ or $b_k$ equals zero, we have the claimed result.
\end{proof}

Now we  state the main theorem in this section showing it is possible to use the above scheme to achieve  non-trivial rate pairs satisfying $R_1+R_2=C_{sum}:=\frac{1}{2}\log(1+h_1^2P+h_2^2P)$.  Furthermore, we show that the whole capacity region is achievable under certain conditions on $h_1, h_2$ and $P$.

\begin{theorem}[Capacity achieving with CFMA]
We consider the two-user Gaussian MAC in (\ref{eq:system_clean_MAC}) where two  sums with coefficients $\ve a$ and $\ve b$ are decoded. We assume that $a_k\neq 0$ for $k=1,2$ and define 
\begin{align}
A:=\frac{h_1h_2P}{\sqrt{1+h_1^2P+h_2^2P}}.
\end{align}

\textbf{Case I):} If it holds that 
\begin{align}
A< 3/4,
\label{eq:AchieCap_no}
\end{align}
the sum capacity cannot be achieved by the proposed coding scheme.

\textbf{Case II):} If it holds that
\begin{align}
A\geq  3/4,
\label{eq:AchieCap_part}
\end{align}
the sum rate capacity can be achieved by decoding two integer sums using $\ve a=(1,1), \ve b=(0,1)$ with message  rate pairs 
\begin{align}
&R_1=r_1(\ve a, \beta_2), R_2=r_2(\ve b|\ve a,\beta_2) \label{eq:rate_b01}\\
&\mbox{ with some } \beta_2\in[\beta_2',\beta_2''] \nonumber 
\end{align}
or using $\ve a=(1,1), \ve b=(1,0)$ with  message rate pairs
\begin{align}
&R_1=r_1(\ve b|\ve a, \beta_2), R_2=r_2(\ve a,\beta_2) \label{eq:rate_b10}\\
 &\mbox{ with some }  \beta_2\in[\beta_2',\beta_2'']\nonumber
\end{align}
where $\beta_2',\beta_2''$ are two real roots of the quadratic equation 
\begin{align}
f(\beta_2):=K(\ve a,\beta_2)-\beta_2\sqrt{1+h_1^2P+h_1^2P}=0.
\label{eq:mac_achive_inequa}
\end{align}
The expressions $r_k(\ve a, \beta_2)$, $r_k(\ve b|\ve a, \beta_2)$and $K(\ve a,\beta_2)$ are given in (\ref{eq:rate_a_clean}),  (\ref{eq:rate_b_clean}) and (\ref{eq:contant_K}) by setting $\beta_1=1$, respectively.

\textbf{Case III:} If it holds that
\begin{align}
A\geq  1,
\label{eq:AchieCap_whole}
\end{align}
by choosing $\ve a=(1,1)$ and $\ve b=(0,1)$ or $\ve b=(1,0)$, the  achievable rate pairs in (\ref{eq:rate_b01}) and (\ref{eq:rate_b10}) cover the whole dominant face of the capacity region.

\label{thm:clean_MAC}
\end{theorem}


\begin{remark}
Figure \ref{fig:partition} shows the achievability of our scheme for different values of received signal-to-noise ratio $h_k^2P$. In Region III (a sufficient condition is  $h_k^2P\geq 1+\sqrt{2}$ for $k=1, 2$), we can achieve any rate pair in the capacity region. In Region I and II the proposed scheme is not able to achieve the entire region. However, we should point out that if we choose the coefficients to be  $\ve a=(1,0), \ve b=(0,1)$ or $\ve a=(0,1), \ve b=(1,0)$, the CFMA scheme reduces to the conventional successive cancellation decoding, and is \textit{always} able to achieve the corner point of the capacity region, \textit{irrespective} of the values of $h_1, h_2$ and $P$.
\end{remark}

\begin{proof}
It is easy to see from the rate expressions (\ref{eq:rate_a_clean}) and (\ref{eq:rate_b_clean}) that we can without loss of generality assume $\beta_1=1$ in the following derivations.  We do not consider the case when $a_k=0$ for $k=1$ or $k=2$, which is just the classical successive cancellation decoding.  Also notice that it holds:
\begin{align}
r_1(\ve a,\beta_2)+r_2(\ve b|\ve a,\beta_2)&=r_2(\ve a,\beta_2)+r_1(\ve b|\ve a,\beta_2)\nonumber\\
&=\frac{1}{2}\log\frac{1+(h_1^2+h_2^2)P}{(a_2b_1-a_1b_2)^2} \nonumber \\
&=C_{sum}-\log|a_2b_1-a_1b_2|
\label{eq:sum_of_rates}
\end{align}
We start with \textbf{Case I)} when the sum capacity cannot be achieved. This happens when 
\begin{align*}
r_k(\ve a,\beta_2)<r_k(\ve b|\ve a,\beta_2), k=1,2
\end{align*}
for  \textit{any} choice of $\beta_2$, which is equivalent to
\begin{align*}
f(\beta_2)>0
\end{align*}
where $f(\beta_2)$ is given in (\ref{eq:mac_achive_inequa}). To see this, notice that Theorem \ref{thm:MAC_general} implies that in this case the sum message rate is
\begin{align*}
R_1+R_2=r_1(\ve a,\beta_2)+r_2(\ve a,\beta_2)
\end{align*}
for $a_k\neq 0$. Due to Eqn. (\ref{eq:sum_of_rates}) we can upper bound the sum message rate by
\begin{align*}
R_1+R_2&<r_1(\ve a)+r_2(\ve b|\ve a,\beta_2)\leq C_{sum}\\
R_1+R_2&<r_2(\ve a)+r_1(\ve b|\ve a,\beta_2)\leq C_{sum},
\end{align*}
meaning the sum capacity is not achievable. It remains to  characterize the condition under which the inequality $f(\beta_2)>0$ holds. It is easy to see the expression $f(\beta_2)$ is a quadratic function of $\beta_2$ with the leading coefficient $a_2^2(1+h_1^2P)$. Hence $f(\beta_2)>0$ always holds if the equation $f(\beta_2)=0$ does not have any real root. The solutions of $f(\beta_2)=0$ are given by
\begin{subequations}
\begin{align}
\beta_2':=\frac{2a_1a_2h_1h_2P+S-\sqrt{SD}}{2(a_2^2+a_2^2h_1^2P)}\\
\beta_2'':=\frac{2a_1a_2h_1h_2P+S+\sqrt{SD}}{2(a_2^2+a_2^2h_1^2P)}
\end{align}
\label{eq:roots}
\end{subequations}
with 
\begin{align*}
S&:=\sqrt{1+(h_1^2+h_2^2)P}\\
D&:=S(1-4a_1^2a_2^2)+4Pa_1a_2h_1h_2
\end{align*}
Inequality $f(\beta_2)>0$ holds for all real $\beta_2$ if $D<0$ or equivalently
\begin{align}
\frac{h_1h_2P}{\sqrt{1+(h_1^2+h_2^2)P}}<\frac{4a_1^2a_2^2-1}{4a_1a_2}
\label{eq:case1_general}
\end{align}
The R.H.S. of the above inequality is minimized by choosing $a_1=a_2=1$ which yields the condition (\ref{eq:AchieCap_no}). This is shown in Figure \ref{fig:clean_caseI}: in this case the computation rate pair of the first sum $\tilde{\ve t}_1+\tilde{\ve t}_2$ is too small and it cannot reach the sum capacity.

In \textbf{Case II)} we require $r_k(\ve a,\beta_2)\geq r_k(\ve b|\ve a,\beta_2)$ or equivalently $f(\beta_2)\leq 0$ for some $\beta_2$.  By the derivation above, this is possible if $D\geq 0$ or equivalently
\begin{align}
\frac{h_1h_2P}{\sqrt{1+(h_1^2+h_2^2)P}}\geq \frac{4a_1^2a_2^2-1}{4a_1a_2}
\label{eq:cap_achiev_general_a}
\end{align}
If we choose the coefficients to be $\ve a=(a_1,a_2)$ and $\ve b=(0, b_2)$ for some nonzero integers $a_1,a_2,b_2$, Theorem \ref{thm:MAC_general} implies the sum rate is
\begin{align*}
R_1+R_2&=r_1(\ve a,\beta_2)+r_2(\ve b|\ve a,\beta_2)\\
&=C_{sum}-\log|a_2b_1-a_1b_2|
\end{align*}
If the coefficients satisfy $|a_2b_1-a_1b_2|=1$, the sum capacity is achievable by choosing $\beta_2\in[\beta_2',\beta_2'']$, with which the inequality (\ref{eq:cap_achiev_general_a}) holds. Notice that if we choose $\beta_2\notin [\beta_2',\beta_2'']$, then $r_k(\ve a,\beta_2)< r_k(\ve b|\ve a,\beta_2)$ and we are back to Case I). The condition $|a_2b_1-a_1b_2|=1$ is satisfied if the coefficients are chosen to be $\ve a=(1,1), \ve b=(0,1)$. For simplicity we collect these two vectors and denote them as $\ve A_1:=(\ve a^T,  \ve b^T)^T$.

The same result holds if the coefficients are of the form $\ve a=(a_1,a_2), \ve b=(b_1, 0)$ and in particular $\ve a=(1,1), \ve b=(1,0)$. Similarly we denote these two vectors using $\ve A_2:=(\ve a^T,  \ve b^T)^T$. We will let the coefficients be $\ve A_1$ or $\ve A_2$ for now and comment on other choices of coefficients later. With this choice of $\ve a$ the inequality (\ref{eq:cap_achiev_general_a}) is just the condition (\ref{eq:AchieCap_part}).

In general, not the whole dominant face of the capacity region can be achieved by varying $\beta_2\in[\beta_2',\beta_2'']$.  One important choice of $\beta_2$ is $\beta_2^{(1)}:=\frac{h_1h_2P}{1+h_1^2P}$. With this choice of $\beta_2$ and coefficients $\ve A_1$ we have
\begin{align}
R_1&=r_1(\ve a, \beta_2^{(1)})=\frac{1}{2}\log(1+h_1^2P)\\
R_2&=r_2(\ve b|\ve a,\beta_2^{(1)})=\frac{1}{2}\log(1+\frac{h_2^2P}{1+h_1^2P})
\end{align}
which is one corner point of the capacity region. Similarly with $\beta_2^{(2)}:=\frac{1+h_2^2P}{h_1h_2P}$ and coefficients $\ve A_2$ we have
\begin{align}
R_2&=r_2(\ve a, \beta_2^{(2)})=\frac{1}{2}\log(1+h_2^2P)\\
R_1&=r_1(\ve b|\ve a,\beta_2^{(2)})=\frac{1}{2}\log(1+\frac{h_1^2P}{1+h_2^2P})
\end{align}
which is another corner point of the capacity region. If the condition $\beta_2^{(1)},\beta_2^{(2)}\notin[\beta_2',\beta_2'']$ is not fulfilled, we  cannot choose $\beta_2$ to be $\beta_2^{(1)}$ or $\beta_2^{(2)}$ hence  cannot achieve the corner points of the capacity region. In Figure \ref{fig:clean_caseII} we give an example in this case where only part of rate pairs on the dominant face can be achieved.

In \textbf{Case III)} we require $\beta_2^{(1)},\beta_2^{(2)}\in[\beta_2',\beta_2'']$.  In Appendix \ref{sec:appen_proof_clean} we show that $\beta_2^{(1)},\beta_2^{(2)}\in[\beta_2',\beta_2'']$ if and only if the condition (\ref{eq:AchieCap_whole}) is satisfied.  With the coefficients $\ve A_1$, the achievable rate pairs $(r_1(\ve a,\beta_2),r_2(\ve b|\ve a,\beta_2))$ lies on the dominant face by varying $\beta_2$ in the interval $[\beta_2^{(1)},\beta_2'']$ and in this case we do not need to choose $\beta_2$ in the interval $[\beta_2',\beta_2^{(1)})$, see Figure \ref{fig:clean_caseIII01} for an example. Similarly with  coefficients $\ve A_2$, the achievable rate pairs $(r_1 (\ve b|\ve a,\beta_2),r_2(\ve a,\beta_2))$ lie on the dominant face by varying $\beta_2$ in the interval $[\beta_2',\beta_2^{(2)}]$ and we do not need to let $\beta_2$ take values in the interval $(\beta_2^{(2)},\beta_2'']$, see Figure \ref{fig:clean_caseIII10} for an example. Since we always have $r_1(\ve a,\beta_2')\geq r_1(\ve b|\ve a,\beta_2'')$ and $r_2(\ve b|\ve a,\beta_2')\geq r_2(\ve a,\beta_2'')$, the achievable rate pairs with coefficients $\ve A_1$ and $\ve A_2$ cover the whole dominant face of the capacity region.

\begin{figure}[!h]
\centering
\includegraphics[scale=0.5]{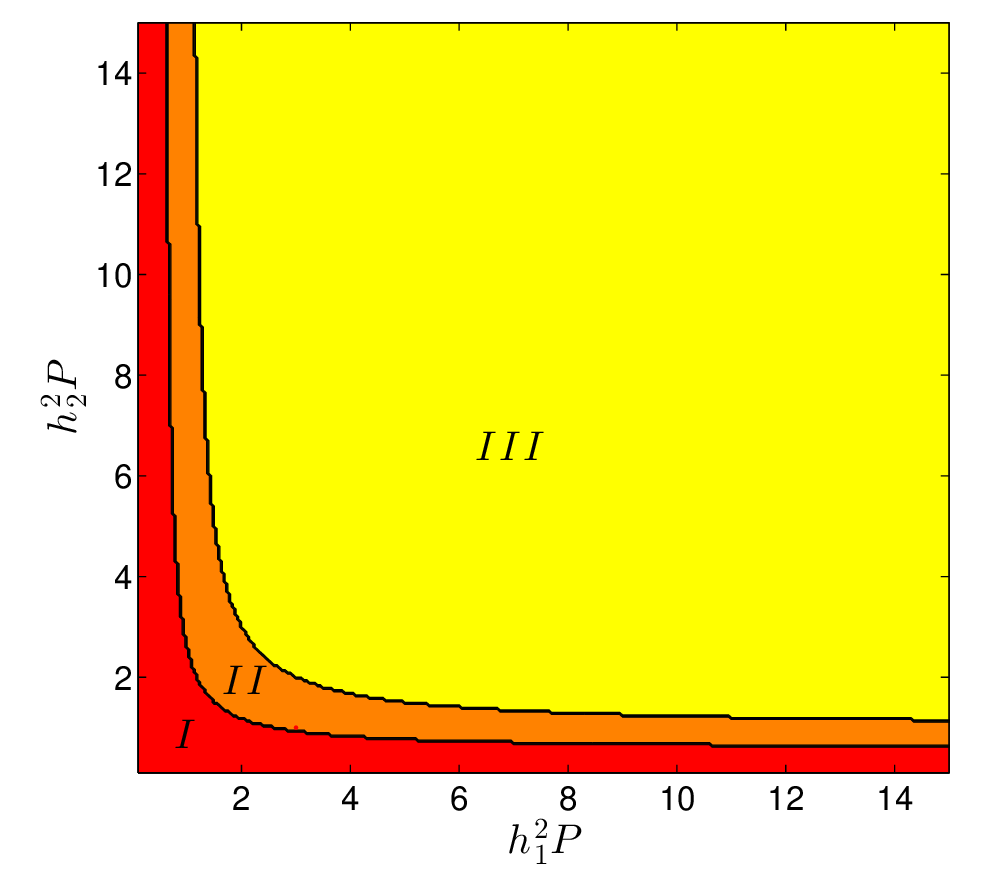}
\caption{The plane of the received SNR $h_1^2P, h_2^2P$ is divided into three regions. Region I corresponds to Case I when the condition (\ref{eq:AchieCap_no}) holds and the scheme cannot achieve points on the boundary of the capacity region. In Region II the condition  (\ref{eq:AchieCap_part})  is  met but the condition (\ref{eq:AchieCap_whole}) is not,  hence only  part of the points on the capacity boundary can be achieved.  Region III corresponds to Case III where (\ref{eq:AchieCap_whole}) are satisfied and the proposed scheme can achieve any point in the capacity region.}
\label{fig:partition}
\end{figure}

\begin{figure*}[!h] \centerline{\subfloat[Case I with $h_1=1, h_2=\sqrt{2}, P=1$]{\includegraphics[width
=2.8in]{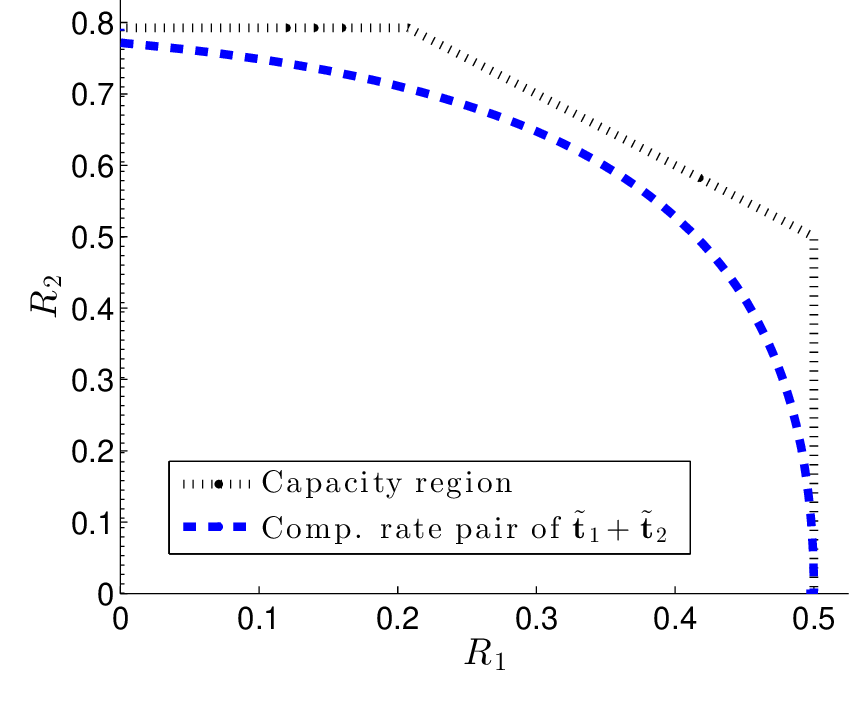} \label{fig:clean_caseI}} \hfil \subfloat[Case II with $h_1=1, h_2=\sqrt{2}, P=1.2$] {\includegraphics[width=2.8in]{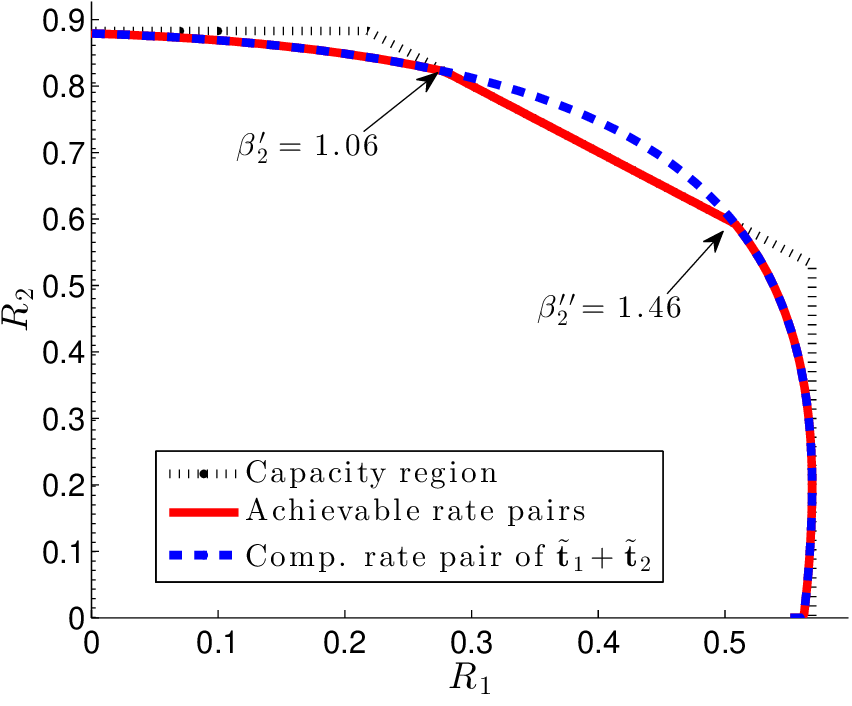}
\label{fig:clean_caseII}}} \caption{Plot (a) shows the achievable rate pairs in Case I. In this case the condition (\ref{eq:AchieCap_no}) is satisfied and the computation rate pair of the first sum is too small. It has no intersection with the dominant face hence cannot achieve sum rate capacity. Notice that the (message) rate pairs contained in the computation rate region are achievable. Plot (b) shows the situation in Case II. In this case the condition (\ref{eq:AchieCap_part}) is fulfilled and the computation rate pair of the first sum is larger. It intersects with the dominant face hence the sum capacity is achievable. In this example the condition (\ref{eq:AchieCap_whole}) is not satisfied hence only part of the dominant face can be achieved, as depicted in the plot. The rate pair segement on the dominant face can be achieved by choosing $\ve a=(1,1)$, $\ve b=(1,0)$ or $\ve b=(0,1)$ and varying $\beta_2\in[\beta_2',\beta_2'']$. Choosing $\beta_2$ to be $\beta_2',\beta_2''$ gives the end points of the segement. We emphesize that if we choose $\ve a=(1,0), \ve b=(0,1)$ or $\ve a=(0,1), \ve b=(1,0)$, i.e., the conventional successive cancellation decoding, we can always achieve the whole capacity region, irrespective of the condition (\ref{eq:AchieCap_no}) or (\ref{eq:AchieCap_part}). } \label{fig:clear_caseI_II} 
\end{figure*}

\begin{figure*}[!h] \centerline
{\subfloat[Case III with $h_1=1, h_2=\sqrt{2},P=4$]{\includegraphics[width=2.8in]{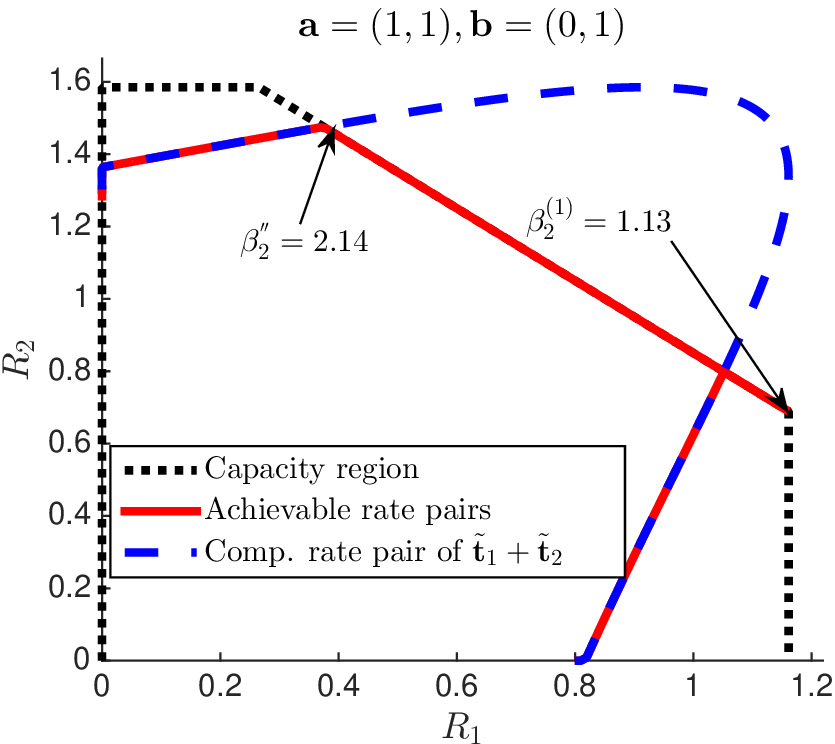} \label{fig:clean_caseIII01}} \hfil 
\subfloat[Case III with $h_1=1, h_2=\sqrt{2},P=4$]{\includegraphics[width=2.5in]{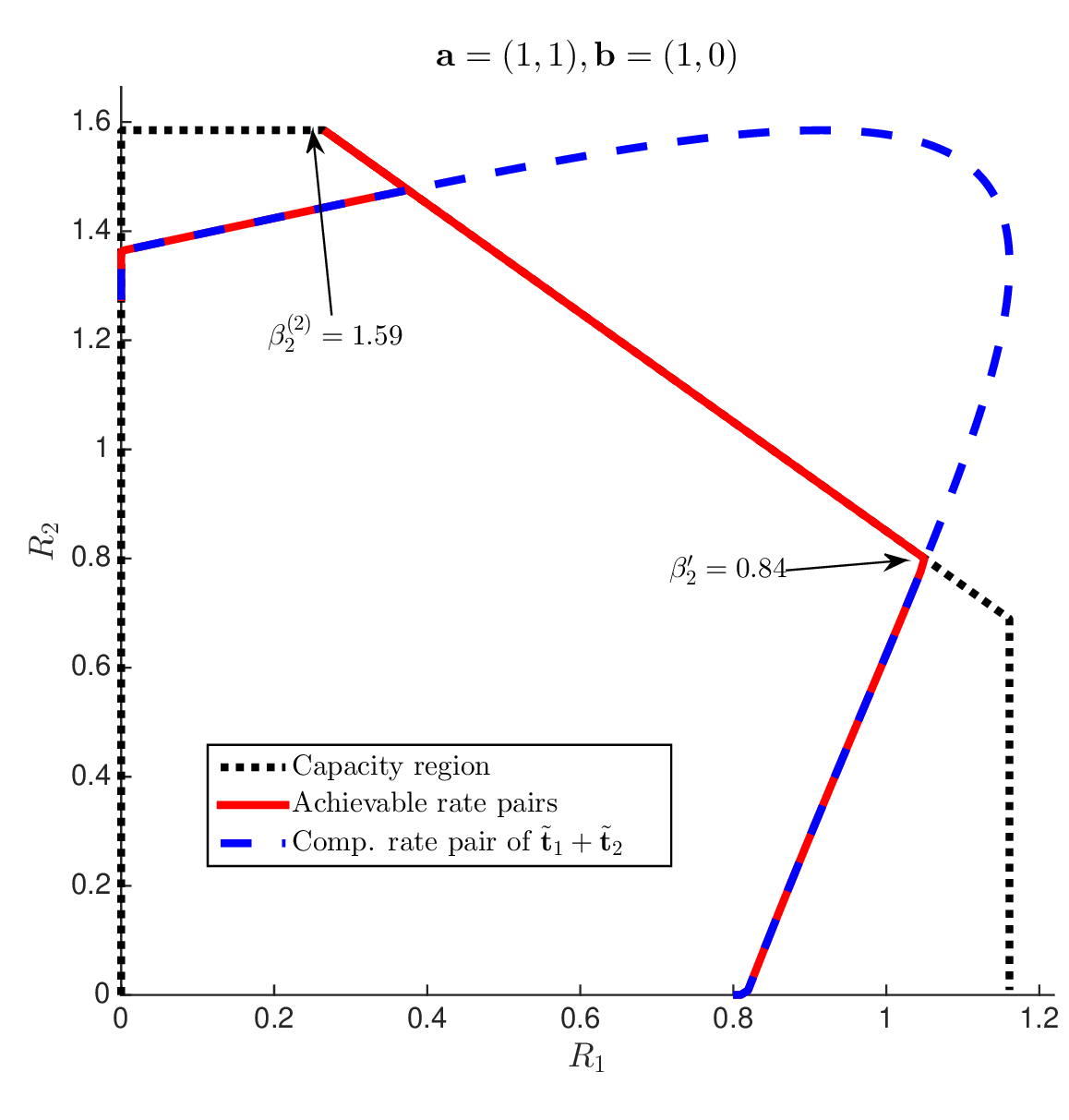}
\label{fig:clean_caseIII10}}} 
\caption{Achievable rate pairs in Case III. (The capacity region and the computation rate pairs in the two plots are the same.) In this case the condition (\ref{eq:AchieCap_whole}) is satisfied hence the computation rate pair of the first sum  is large enough to  achieve the whole capacity region  by decoding two nontrivial integer sums.  Plot (a) shows the achievable rate pairs by choosing $\ve a=(1,1), \ve b=(0,1)$ and varying $\beta_2\in[\beta_2^{(1)},\beta_2'']$.  Plot (b) shows the achievable rate pairs by choosing $\ve a=(1,1), \ve b=(1,0)$ and varying $\beta_2\in[\beta_2',\beta_2^{(2)}]$. The union of the achievable rate pairs with coefficients   cover the whole dominant face of the capacity region.} 
\label{fig:caseIII} \end{figure*}

As mentioned previously, a similar idea is developed in \cite{ordentlich_successive_2013} showing that certain isolated points on the capacity boundary are achievable under certain condition. Before ending the proof, we comment on two main points in the  proposed CFMA scheme, which enable us to improve upon the previous result. The first point is the introduction of the scaling parameters $\beta_k$ which allow us to adjust the rates of the  two users. More precisely, equations (\ref{eq:rate_a_clean_proof}) and (\ref{eq:rate_b_clean_proof}) show that the scaling parameters not only affect the equivalent noise $N_1(\alpha_1)$ and $N_2(\alpha_2,\lambda)$, but also balance the rates of different users (as they also appear in the numerators). We need to adjust the rates of the two users carefully through these parameters to make sure that the rate pairs lie on the capacity boundary.  The second point is that in order to achieve the whole capacity boundary, it is very important to choose the right coefficients of the sum. In particular for the two-user Gaussian MAC, the coefficients for the second sum should be $(1,0)$ or $(0,1)$. More discussions on the choice of coefficients are given in the next section.
\end{proof}

\subsection{On the choice of coefficients}
In Theorem \ref{thm:clean_MAC} we only considered the coefficients $\ve a=(1,1)$, $\ve b=(1,0)$ or $\ve b=(0,1)$. It is natural to ask whether choosing other coefficients could be advantageous.   We first consider the case when the coefficients $\ve a$ of the first sum is chosen differently.

\begin{lemma}[Achieving capacity with  a different $\ve a$]
Consider a $2$-user Gaussian MAC where the receiver decodes two integer sums of the codewords with coefficients $\ve a=(a_1, a_2)$ and $\ve b=(0,1)$ or $\ve b=(1,0)$. Certain rate pairs on the dominant face are achievable if it holds that
\begin{align}
\frac{h_1h_2P}{\sqrt{1+(h_1^2+h_2^2)P}}\geq \frac{4a_1^2a_2^2-1}{4a_1a_2}.
\label{eq:no_capacity}
\end{align}
Furthermore the corner points of the capacity region are achievable if it holds that
\begin{align}
\frac{h_1h_2P}{\sqrt{1+(h_1^2+h_2^2)P}}\geq a_1a_2.
\end{align}
\label{lemma:coef}
\end{lemma}
\begin{proof}
The proof of the first statement is given in the proof of Theorem \ref{thm:clean_MAC}, see Eqn. (\ref{eq:case1_general}). The proof of the second statement is omitted as it is the same as the proof of Case III in Theorem \ref{thm:clean_MAC} with a general $\ve a$. 
\end{proof}

This result suggests that although it is always possible to achieve the sum capacity with any $\ve a$, provided that the SNR of users are large enough, the choice $\ve a=(1,1)$ is the best, in the sense that it requires the lowest SNR threshold, above which the sum capacity or the whole capacity region  is achievable.  

To illustrate this, let us reconsider the setting of Fig. \ref{fig:caseIII}, but  select  coefficients $\ve a$ different from $(1,1)$. As can be seen in Figure \ref{fig:coef1221_noCapacity}, it is not possible to achieve the sum capacity with $\ve a=(1,2)$ or $\ve a=(2,1)$. If we increase the power from $P=4$ to $P=10$, part of the capacity boundary is achieved, as shown in Figure \ref{fig:coef1221}.  However in this case we cannot achieve the whole capacity region. The reason lies in the fact that the computation rate pairs are different for $\ve a=(1,2)$ and $\ve a=(2,1)$.

\begin{figure*}[!h] \centerline
{\subfloat[Achievable (computation) rate pairs with $h_1=1, h_2=\sqrt{2},P=4$ and $\ve a=(1,2)$ or $\ve a=(2,1)$]{\includegraphics[width=2.8in]{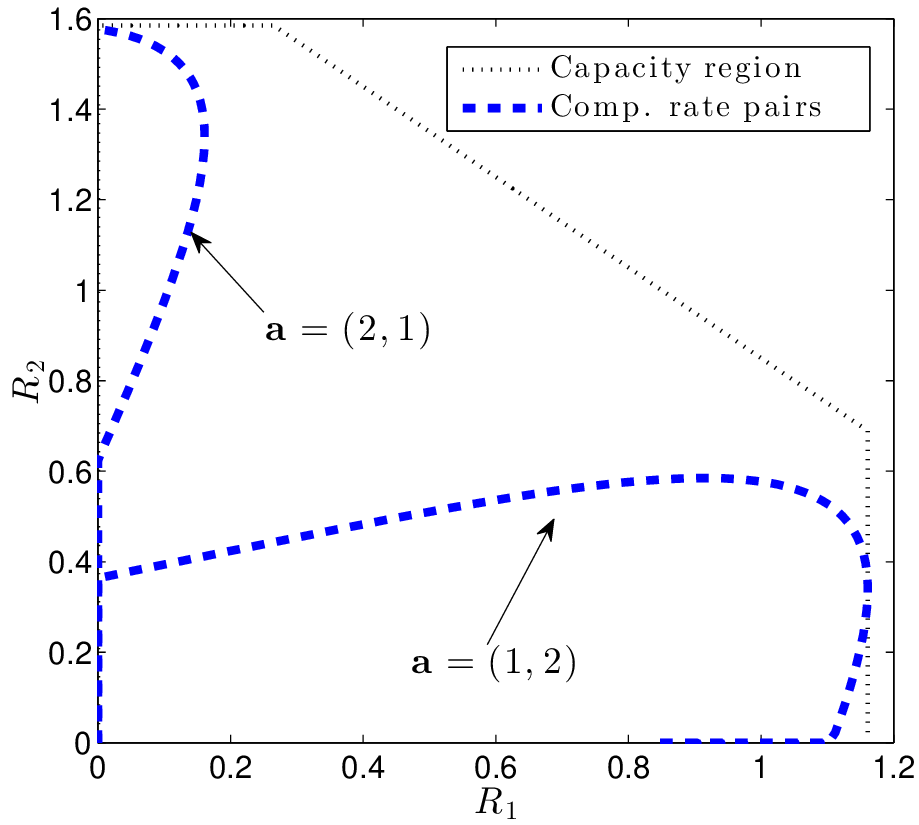} \label{fig:coef1221_noCapacity}} \hfil 
\subfloat[Achievable rate pairs with $h_1=1, h_2=\sqrt{2},P=10$ and $\ve a=(1,2)$ or $\ve a=(2,1)$. ]{\includegraphics[width=2.8in]{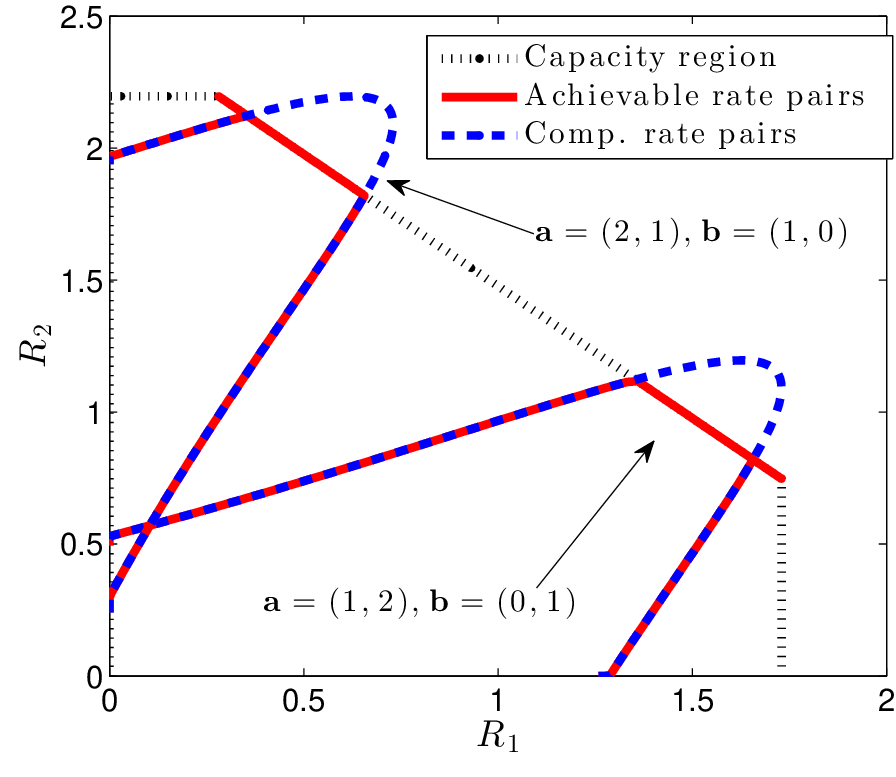}
\label{fig:coef1221}}} 
\caption{In the left plot we show the computation rate pairs with parameters $h_1=1, h_2=\sqrt{2}, P=4$ where the coefficients of the first sum are chosen to be $\ve a=(1,2)$ or $\ve a=(2,1)$. In this case the condition (\ref{eq:no_capacity}) is not satisfied hence no point on the dominant face can be achieved for the first sum. Compare it to the example in Figure \ref{fig:clean_caseIII01} or \ref{fig:clean_caseIII10} where $\ve a=(1,1)$ and the whole capacity region is achievable. We also note that the achievable computation rate pairs depicted in the Figure are also achievable message rate pairs, which can be shown using Theorem \ref{thm:MAC_general}.  In the right plot we show the achievable rate pairs with parameters $h_1=1, h_2=\sqrt{2}, P=10$ where the coefficient of the first sum is chosen to be $\ve a=(1,2)$ or $\ve a=(2,1)$. In this case we can achieve the sum capacity but cannot obtain the whole dominant face. In contrast, choosing $\ve a=(1,1)$ achieves the whole dominant face.} 
 \end{figure*}

Now we consider a different choice on the coefficients $\ve b$ of the second sum. Although from the perspective of solving equations,  having two sums with coefficients $\ve a=(1,1), \ve b=(1,0)$ or $\ve a=(1,1), \ve b=(1,2)$ is equivalent, here it is very important to choose $\ve b$ such that it has one zero entry. Recall the result in Theorem \ref{thm:MAC_general} that if $b_k\neq 0$ for $k=1,2$, then both message rates $R_1, R_2$ will have two constraints from the two sums we decode. This extra constraint will diminish the achievable rate region, and in particular it only achieves some isolated points on the dominant face. This is illustrated by the example in Figure \ref{fig:b21}.

As a rule of thumb, the receiver should always decode the sums whose coefficients are as small as possible in CFMA.

\begin{figure}[!h]
\centering
\includegraphics[width=3.0in]{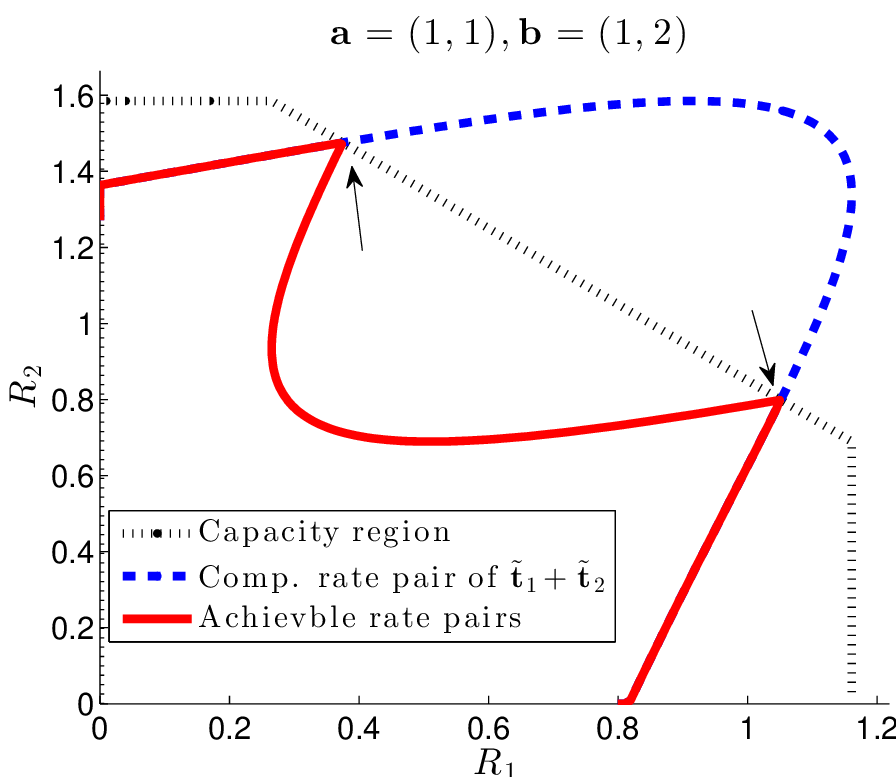}
\caption{The achievable rate pairs with parameters $h_1=1, h_2=\sqrt{2}, P=4$. In this case the condition (\ref{eq:AchieCap_whole}) is satisfied hence the first sum is chosen properly. But as we choose $\ve b=(1,2)$, only two isolated points (indicated by arrows) on the dominant face can be achieved. This is due to the fact non-zero entries in $\ve b$ will give an extra  constraint on the rate, cf. Theorem \ref{thm:MAC_general}. Compare it with the example in Figure \ref{fig:clean_caseIII10}.} 
\label{fig:b21}
\end{figure}

\subsection{A comparison with other multiple access techniques}
The CFMA strategy provides an alternative to  existing multiple-access techniques. In this section we  lay out the  limitations and possible advantages of the  CFMA scheme, and  compare it with other existing multiple access techniques. 
\begin{itemize}
\item We have mentioned that one advantage of  CFMA scheme is that the decoder used for  lattice decoding is a single-user decoder,  combined with the successive cancellation. Compared to a MAC decoder with joint-decoding, it permits a simpler receiver architecture. In other words, a lattice codes decoder for a point-to-point Gaussian channel can be directly used for a Gaussian MAC with a simple modification. In contrast a joint-decoder needs to perform estimations simultaneously on both messages hence generally has higher complexity. 
\item Compared to the successive cancellation decoding scheme with time sharing, CFMA also performs successive cancellation decoding but does not require time-sharing for achieving the desired rate pairs in the capacity region (provided that the mild condition on SNR is fulfilled).
\item The rate-splitting scheme also permits a single-user decoder at the receiver. As shown in \cite{rimoldi_rate-splitting_1996}, $2K-1$ single-user decoders are enough for the rate-splitting scheme in a $K$-user Gaussian MAC.  On the other hand, CFMA requires a matrix inversion operation to solve individual messages after collecting different sums which could be computationally expensive. However as shown in an example in Section \ref{sec:symmetric_K_MAC}, we can often choose the matrix to have very special structure and make it very easy to solve for individual messages.  Furthermore, CFMA can also be combined with the rate-splitting technique (i.e. decoding integer sums of the split messages), although this is not necessary for the multiple access problem considered in this paper.
\item More importantly,   CFMA is able to achieve the optimal rate pairs in certain communication scenarios while the conventional single-user decoding with time-sharing or the rate splitting technique fails.  An example for such scenario is the Gaussian interference channel with strong interference and detailed discussions are given in \cite{ZhuGastpar_CFMA}.
\end{itemize}

\section{The K-user Gaussian MAC}\label{sec:K_user}
In this section we consider the general $K$-user Gaussian MAC given in (\ref{eq:sys_K_MAC}). Continuing with the coding scheme for the $2$-user Gaussian MAC, in this case the receiver decodes $K$ integer sums with linearly independent coefficients and uses them to solve for the individual messages. The coefficients of the $K$ sums will be denoted by a \textit{coefficient matrix} $\ve A\in\mathbb Z^{K\times K}$
\begin{align}
\ve A:=(\ve a_1^T \ldots \ve a_K^T)^T=
\begin{pmatrix}
a_{11} &a_{12} &\ldots &a_{1K}\\
a_{22} &a_{22} &\ldots &a_{2K}\\
\ldots &\ldots &\ldots &\ldots\\
a_{K1} &a_{K2} &\ldots &a_{KK}\\
\end{pmatrix}
\end{align}
where the row vector $\ve a_{\ell}:=(a_{\ell 1},\ldots,a_{\ell K})\in \mathbb Z^{1\times K}$ denotes the coefficients of the $\ell$-th sum, $\sum_{k=1}^Ka_{\ell k}\tilde{\ve t}_k$.

The  following theorem gives an achievable message rate tuple for the general $K$-user Gaussian MAC. It is an extension of  \cite[Thm. 2]{ordentlich_successive_2013} as the scaling parameters $\beta_k$ in our proposed scheme allow a larger achievable rate region. 

\begin{theorem}[Achievability for the $K$-user Gaussian MAC]
Consider the $K$-user Gaussian MAC in (\ref{eq:sys_K_MAC}). Let $\ve A$ be a full-rank integer matrix and $\beta_1,\ldots,\beta_K$ be $K$ non-zero real numbers. We define $\ve B:=\text{diag}(\beta_1,\ldots,\beta_K)$ and 
\begin{align}
\ve K_{\ve Z'}:=P\ve A\ve B(\ve I+P\ve h\ve h^T)^{-1}\ve B^T\ve A^T
\label{eq:cov_noise}
\end{align}
Let the matrix $\ve L$ be the unique Cholesky factor of the matrix $\ve A\ve B(\ve I+P\ve h\ve h^T)^{-1}\ve B^T\ve A^T$, i.e.
\begin{align}
\ve K_{\ve Z'}=P\ve L\ve L^T
\label{eq:chol}
\end{align}  
The message rate tuple $(R_1,\ldots, R_K)$ is achievable with
\begin{align*}
R_k= \min_{\ell\in[1:K]}\left\{\frac{1}{2}\log^+\left(\frac{\beta_k^2}{L_{\ell\ell}^2}\right)\cdot\chi(a_{\ell k})\right\},  k=1,\ldots, K
\end{align*}
where we define
\begin{align}
\chi(x)=\begin{cases}
+\infty  &\mbox{if } x=0,\\
1	&\mbox{otherwise.}
\end{cases}
\end{align}
Furthermore if $\ve A$ is a unimodular ($|\ve A|=1$) and  $R_k$ is of the form
\begin{align}
R_k=\frac{1}{2}\log\left(\frac{\beta_k^2}{L_{\Pi(k)\Pi(k)}^2}\right), k=1,\ldots, K
\end{align}
for some permutation $\Pi$ of the set $\{1,\ldots, K\}$, then the sum rate satisfies
\begin{align}
\sum_{k=1}^K R_k=C_{sum}:=\frac{1}{2}\log\left(1+\sum_{k=1}^K h_k^2P\right)
\end{align}
\label{thm:K_MAC}
\end{theorem}

\begin{proof}
To prove this result, we will adopt a more compact representation and follow the proof technique given in \cite{ordentlich_successive_2013}. We rewrite the system in (\ref{eq:sys_K_MAC}) as
\begin{align}
\ve Y=\ve h\ve X+\ve z
\end{align}
with $\ve h=(h_1,\ldots, h_K)\in\mathbb R^{1\times K}$ and $\ve X=(\ve x_1^T \ldots  \ve x_K^T)^T\in\mathbb R^{K\times n}$ where each $\ve x_k\in R^{1\times n}$ is the transmitted signal sequence of user $k$ given by
\begin{IEEEeqnarray}{rCl}
\ve x_k=[\ve t_k/\beta_k+\ve d_k]\mode\Lambda_k/\beta_k
\end{IEEEeqnarray}
Similar to the derivation for the $2$-user case, we multiply the channel output by a matrix $\ve F\in\mathbb R^{K\times 1}$ and it can be shown that the following equivalent output can be obtained
\begin{IEEEeqnarray}{rCl}
\tilde{\ve Y}=\ve A\ve T+\tilde{\ve Z}
\label{eq:equivlent_vec}
\end{IEEEeqnarray}
where $\ve T:=(\tilde{\ve t}_1^T \ldots \tilde{\ve t}_K^T)^T\in\mathbb R^{K\times n}$ and the lattice codeword $\tilde{\ve t}_k\in R^{n\times 1}$ of user $k$ is the same as defined in (\ref{eq:t_tilde}). Furthermore the noise $\tilde{\ve Z}\in\mathbb R^{K\times n}$ is given by
\begin{align}
\tilde{\ve Z}=(\ve F\ve h-\ve A\ve B)\ve X+\ve F\ve z
\end{align}
where  $\ve B:=\text{diag}(\beta_1,\ldots,\beta_K)$. The matrix $\ve F$ is chosen to minimize the variance of the  noise:
\begin{align}
\ve F:=P\ve A\ve B\ve h^T\left(\frac{1}{P}\ve I+\ve h\ve h^T\right)^{-1}
\end{align}

As shown in the proof of \cite[Thm. 5]{NazerGastpar_2011}, when analyzing the lattice decoding for the system given in (\ref{eq:equivlent_vec}), we can  consider the system
\begin{align}
\tilde{\ve Y}=\ve A\ve T+\ve Z'
\label{eq:equivalent_vec_Gauss}
\end{align}
where $\ve Z'\in\mathbb R^{K\times n}$ is the equivalent noise and each row $\ve z_k$ is a $n$-sequence of i.i.d Gaussian random variables $z_k$ for $k=1,\ldots,K$. The covariance matrix of the Gaussians $z_1,\ldots,z_K$ is the same as that of the original noise $\tilde{\ve Z}$ in (\ref{eq:equivlent_vec}). It is easy to show that  the covariance matrix of the equivalent noise $z_1,\ldots,z_K$ is given in Eq. (\ref{eq:cov_noise}).

Now instead of doing the successive interference cancellation as in the $2$-user case, we use an equivalent formulation which is called ``noise prediction" in \cite{ordentlich_successive_2013}.  Because the matrix $\ve A\ve B(\ve I+P\ve h\ve h^T)^{-1}\ve B^T\ve A^T$ is positive  definite, it admits the Cholesky factorization hence the covariance matrix $\ve K_{\ve Z'}$ can be rewritten as 
\begin{align}
\ve K_{\ve Z'}=P\ve L\ve L^T
\end{align}
where $\ve L$ is a lower triangular matrix. 

Using the Cholesky decomposition of $\ve K_{\tilde{\ve Z}}$, the system (\ref{eq:equivalent_vec_Gauss}) can be represented as
\begin{align}
\tilde{\ve Y}&=\ve A\ve T+\sqrt{P}\ve L\ve W \nonumber\\
&=\begin{pmatrix}
a_{11} &a_{12} &\ldots &a_{1K}\\
a_{21} &a_{22} &\ldots &a_{2K}\\
\vdots &\vdots &\vdots &\vdots\\
a_{K1} &a_{K2} &\ldots &a_{KK}\\
\end{pmatrix}
\begin{pmatrix}
\tilde{\ve t}_1\\
\tilde{\ve t}_2\\
\vdots\\
\tilde{\ve t}_K
\end{pmatrix}
\nonumber \\
&+\sqrt{P}
\begin{pmatrix}
L_{11} &0  &0 &\ldots &0\\
L_{21} &L_{22} &0 &\ldots &0\\
\vdots &\vdots &\vdots &\vdots&\vdots\\
L_{K1} &L_{K2} &L_{K3} &\ldots &L_{KK}
\end{pmatrix}
\begin{pmatrix}
\ve w_1\\
\ve w_2\\
\vdots\\
\ve w_K\\
\end{pmatrix}
\label{eq:system_lower}
\end{align}
with $\ve W=[\ve w_1^T,\ldots, \ve w_K^T]\in \mathbb R^{K\times n}$ where $\ve w_i\in\mathbb R^{n\times 1}$ is an $n$-length sequence whose  components are i.i.d. zero-mean white Gaussian random variables  with unit variance. This is possible by noticing that $\sqrt{P}\ve L\ve W$ and $\ve Z'$ have the same covariance matrix.  Now we apply lattice decoding to each row of the above linear system. The first row of the equivalent system in (\ref{eq:system_lower}) is given by
\begin{align*}
\tilde{\ve y}_1:=\ve a_1\ve T+\sqrt{P}L_{11}\ve w_1
\end{align*}
Using lattice decoding, the first integer sum $\ve a_1\ve T=\sum_k a_{1k}\tilde{\ve t}_k$ can be decoded reliably if
\begin{align*}
r_k< \frac{1}{2}\log^+\frac{\beta_k^2P}{PL_{11}^2}=\frac{1}{2}\log^+\frac{\beta_k^2}{L_{11}^2}, k=1,\ldots, K
\end{align*}
Notice that if $a_{1k}$ equals zero, the lattice point $\tilde{\ve t}_k$ does not participate in the sum $\ve a_1\ve T$ hence $r_k$ is not constrained as above.

The important observation is that knowing $\ve a_1\ve T$ allows us to recover the noise term $\ve w_1$ from $\tilde{\ve y}_1$. This ``noise prediction" is equivalent to  the successive interference cancellation, see also \cite{ordentlich_successive_2013}.  Hence we could eliminate the term $\ve w_1$ in the second row of the system (\ref{eq:system_lower}) to obtain
\begin{align*}
\tilde{\ve y}_2:=\ve a_2\ve T+\sqrt{P}L_{22}\ve w_2
\end{align*}
The lattice decoding of $\ve a_2\ve T$ is successful if
\begin{align*}
r_k< \frac{1}{2}\log^+\frac{\beta_k^2P}{PL_{22}^2}=\frac{1}{2}\log^+\frac{\beta_k^2}{L_{22}^2}, k=1,\ldots, K
\end{align*}
Using the same idea we can eliminate all  noise terms $\ve w_1,\ldots,\ve w_{\ell-1}$ when  decode the $\ell$-th sum. Hence the rate constraints on $k$-th user when decoding the sum $\ve a_\ell \ve T$ is given by
\begin{align*}
r_k< \frac{1}{2}\log^+\frac{\beta_k^2P}{PL_{\ell\ell}^2}=\frac{1}{2}\log^+\frac{\beta_k^2}{L_{\ell\ell}^2}, k=1,\ldots, K
\end{align*}
When decoding the $\ell$-th sum,  the constraint on $r_k$ will be active only if the coefficient of $\tilde{\ve t}_k$ is not zero. Otherwise this decoding will not constraint $r_k$. This fact is captured by introducing the $\chi$ function in the statement of the Theorem. This gives the claimed expression.

In the case when the achievable message rate $R_k$ is of the form
\begin{align*}
R_k=\frac{1}{2}\log\left(\frac{\beta_k^2}{L_{\Pi(k)\Pi(k)}^2}\right),
\end{align*}
the sum rate is
\begin{align*}
\sum_{k}R_k&=\sum_k \frac{1}{2}\log\frac{\beta_k^2}{L_{\Pi(k)\Pi(k)}^2}\\
&=\frac{1}{2}\log\prod_k\frac{\beta_k^2}{L_{kk}^2}\\
&=\frac{1}{2}\log\frac{\prod_k\beta_k^2}{|\ve L\ve L^T|}\\
&=\frac{1}{2}\log\frac{\prod_k\beta_k^2}{|\ve A\ve B(\ve I+P\ve h\ve h^T)^{-1}\ve B^T\ve A^T|}\\
&=\frac{1}{2}\log |\ve I+P\ve h\ve h^T|+\frac{1}{2}\log\prod_k \beta_k^2-\log|\ve A|-\frac{1}{2}\log|\ve B^T\ve B|\\
&=\frac{1}{2}\log|\ve I+P\ve h\ve h^T|-\log|\ve A|\\
&=C_{sum}-\log|\ve A|.
\end{align*}
If $\ve A$ is unimodular, i.e., $|\ve A|=1$, the achievable sum rate is equal to the sum capacity.
\end{proof}
\begin{remark}
The theorem says that to achieve the sum capacity, we need $\ve A$ to be unimodular and $R_k$ should have the form $R_k=\frac{1}{2}\log\frac{\beta_k^2}{L_{\Pi(k)\Pi(k)}^2}$, whose validity of course depends  on the choice of $\ve A$. It is difficult to characterize the class of $\ve A$ for which this holds.  In the case when $\ve A$ is upper triangular with non-zero diagonal entries  and $L_{11}^2\leq\ldots\leq L_{KK}^2$, this condition holds and in fact in this case we have $R_k=\frac{1}{2}\log\frac{\beta_k^2}{L_{kk}^2}$. It can be seen that we are exactly in this situation when we study the $2$-user MAC in Theorem \ref{thm:clean_MAC}.
\end{remark}

\subsection{An example of a $3$-user MAC}

It is in general difficult to analytically characterize the achievable rate using our scheme of the $K$-user MAC.  We give an example of a $3$-user MAC in Figure \ref{fig:3user} to help visualize the achievable region. The channel has the form
$\ve y=\sum_{k=1}^3\ve x_k+\ve z$ and the receiver decodes three sums with coefficients of the form
\begin{align}
\ve A=\begin{pmatrix}
1 &1 &1\\
 &\ve e_i\\
 &\ve e_j
\end{pmatrix}
\end{align}
for $i,j=1,2,3$ and $i\neq j$ where $\ve e_i$ is a row vector with $1$ in its $i$-th  and zero otherwise. It is easy to see that there are in total $6$ matrices of this form and they all satisfy $|\ve A|=1$, hence it is possible to achieve the capacity of this MAC according to Theorem \ref{thm:K_MAC}. For power $P=8$, most parts of the dominant face are achievable except for three triangular regions. For smaller power $P=2$, the achievable part of the dominant face shrinks and particularly the symmetric capacity point is not achievable. It can be checked that in this example, no other coefficients will give a larger achievable region. 

Unlike the $2$-user case, even with a large power,  not the whole dominant face can be obtained in this symmetric $3$-user MAC under the proposed scheme. To obtain some intuition why it is the case, we consider one edge of the dominant face indicated by the arrow in Figure \ref{fig:3user_p8}.  If we want to achieve the rate tuple on this edge, we need to decode user $1$  last because $R_1$ attains its maximum. Hence a reasonable choice of the coefficients matrix would be
\begin{align}
\ve A'=\begin{pmatrix}
0 &1 &1\\
0 &1 &0\\
1 &0 &0
\end{pmatrix}
\mbox{ or }
\ve A'=\begin{pmatrix}
0 &1 &1\\
0 &0 &1\\
1 &0 &0
\end{pmatrix}
\end{align}
Namely we first decode two sums to solve  both $\ve t_2$ and $\ve t_3$, and then decode $\ve t_1$ without any interference. When decoding the first two sums, we are effectively dealing with a $2$-user MAC while  treating $\ve t_1$ as noise. The crux  is that with $\ve t_1$ as noise, the signal-to-noise ratio of user $2$ and $3$ are too low, such that computation rate pair cannot reach the dominant face of the effective $2$-user MAC with $\ve t_1$ being treated as noise. This is the same situation as the Case I considered in Theorem \ref{thm:clean_MAC}. In Figure \ref{fig:3user_p8} we also plot the achievable rates with the coefficients $\ve A'$ above on the side face. We see when $R_1$ attains its maximal value,  the achievable rates cannot reach the dominant face, as a reminiscence of the $2$-user example in Figure \ref{fig:clean_caseI}.

\begin{figure*}[!thb] \centerline{\subfloat[$h_k=1, P=8$]{\includegraphics[width
=3in]{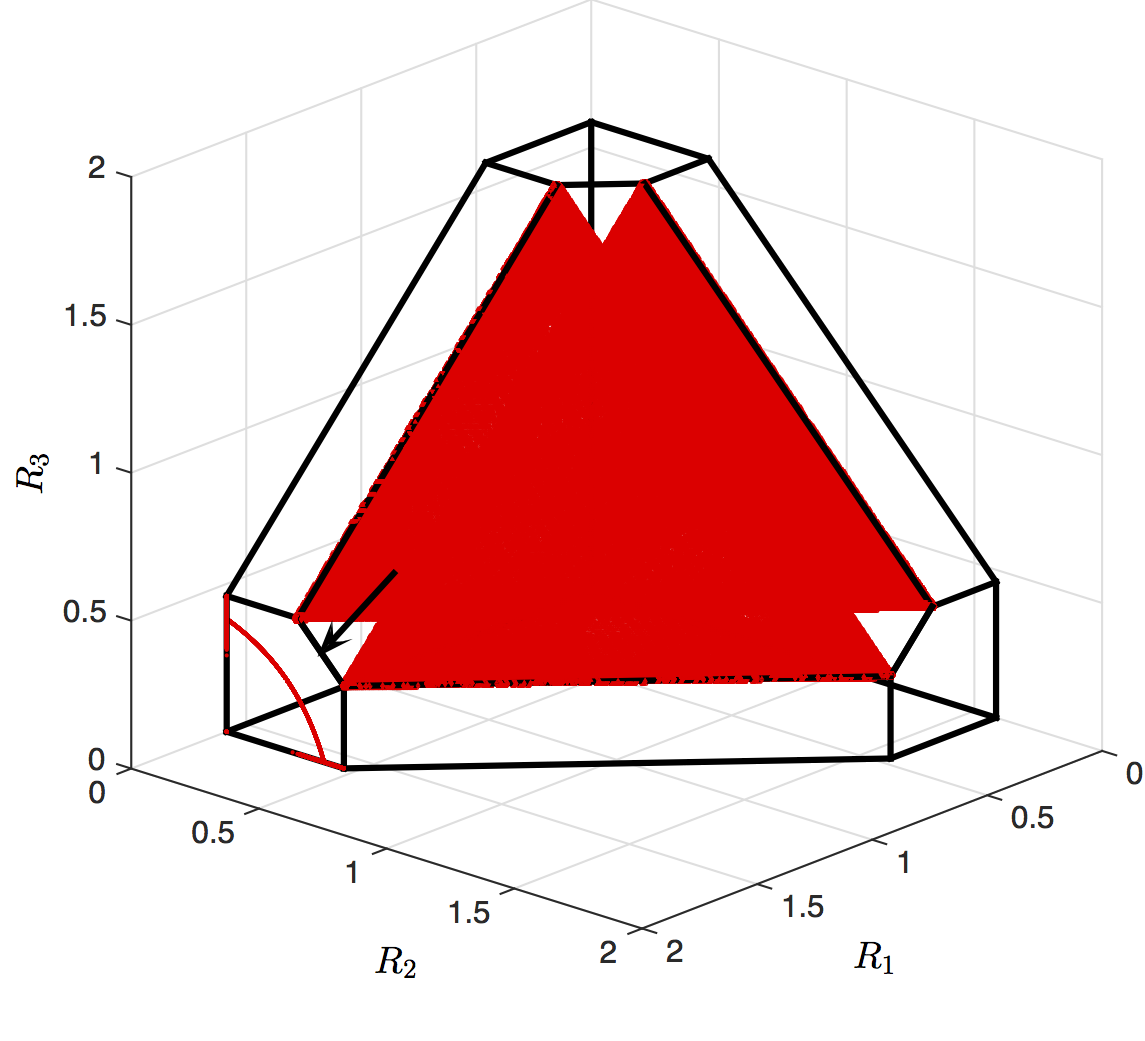} \label{fig:3user_p8}} \hfil \subfloat[$h_k=1, P=2$] {\includegraphics[width=3in]{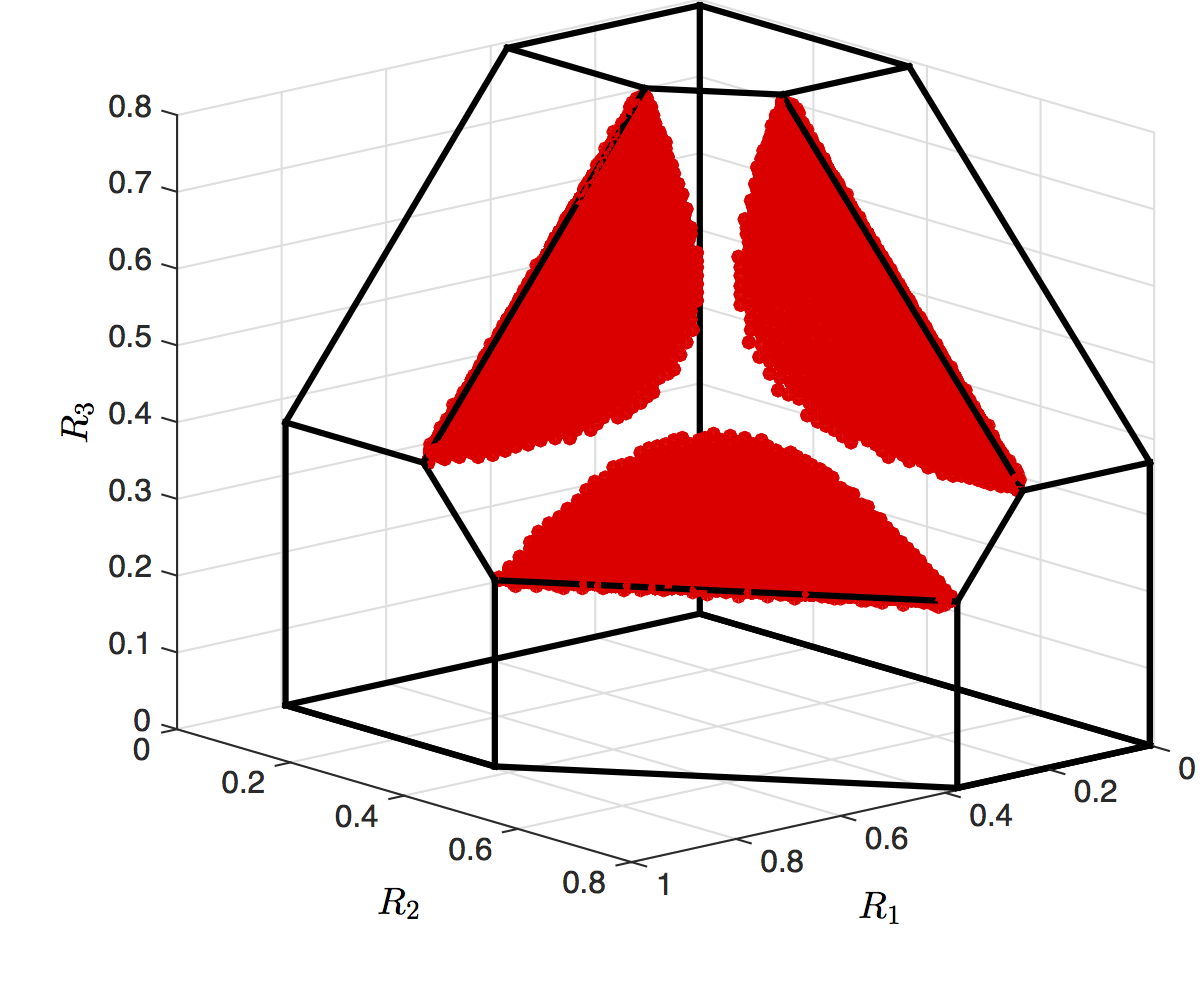}
\label{fig:3user_p2}}} \caption{The achievable rate region (red part) in Theorem \ref{thm:K_MAC} for a symmetric $3$-user Gaussian MAC with $h_k=1$ for $k=1,2,3$ and different powers $P$.} \label{fig:3user} 
\end{figure*}

\subsection{The symmetric capacity for the symmetric Gaussian MAC}\label{sec:symmetric_K_MAC}
As it is  difficult to obtain a complete description of the achievable rate region for a $K$-user MAC, in this section we investigate the simple symmetric channel where all the channel gains are the same.  In this case we can absorb the channel gain into the power constraint and assume without loss of generality the channel model to be
\begin{align*}
\ve y=\sum_{k=1}^K\ve x_k+\ve z
\end{align*}
where the transmitted signal $\ve x_k$ has an average power constraint $P$. We want to see if CFMA  can achieve the symmetric capacity
\begin{align*}
C_{sym}=\frac{1}{2K}\log(1+KP).
\end{align*}
For this specific goal, we will fix our coefficient matrix to be
\begin{align}
\ve A:=\begin{pmatrix}
1 &1 &1 &\ldots &1\\
0 &1 &0 &\ldots &0\\
0 &0 &1 &\ldots &0\\
\vdots &\vdots &\vdots &\ddots &\vdots\\
0  &0 &0 &0 &1
\end{pmatrix}
\label{eq:coef_sym}
\end{align}
Namely we first decode a sum involving all codewords $\sum_{k}\ve t_k$, then decode the individual codewords one by one. Due to symmetry the order of the decoding procedure is irrelevant and  we fix it to be $\ve t_2,\ldots,\ve t_K$. As shown in Theorem \ref{thm:K_MAC}, the analysis of this problem is closely connected to the Cholesky factor  $\ve L$ defined in (\ref{eq:chol}). This connection can be made more explicit if we are interested in the symmetric capacity for the symmetric channel. 

We define 
\begin{align}
\ve C:=\begin{pmatrix}
1 &\beta_2 &\beta_3 &\ldots &\beta_{K}\\
0 &1 &0 &\ldots &0\\
0 &0 &1 &\ldots &0\\
\vdots &\vdots &\vdots &\dots &\vdots\\
0  &0 &0 &0 &1
\end{pmatrix}
\end{align}
and $\ve E$ to be the all-one matrix. Let the lower triangular matrix $\tilde{\ve L}$ denote the unique Cholesky factorization of the matrix $\ve C(\ve I-\frac{P}{1+KP}\ve E)\ve C^T$, i.e., 
\begin{align}
\ve C\left(\ve I-\frac{P}{1+KP}\ve E\right)\ve C^T=\tilde{\ve L}\tilde{\ve L}^T.
\label{eq:L_tilde}
\end{align}

\begin{proposition}[Symmetric capacity]
If there exist real numbers $\beta_2,\ldots,\beta_K\geq 1$ with $|\beta_k|\geq 1$ such that the diagonal entries of  $\tilde{\ve L}$ given in (\ref{eq:L_tilde}) are equal in amplitude i.e., $|\tilde{L}_{kk}|=|\tilde{L}_{jj}|$ for all $k,j$, then the symmetric capacity, i.e., $R_k=C_{sym}$ for all $k$, is achievable for the symmetric $K$-user Gaussian MAC. 
\label{prop:C_sym}
\end{proposition}
\begin{proof}
Recall we have $\ve B=\text{diag}(\beta_1,\beta_2,\ldots,\beta_K)$. Let $\ve A$ be as given in (\ref{eq:coef_sym}) and the channel coefficients $\ve h$  be the all-one vector. Substituting them into (\ref{eq:cov_noise}), (\ref{eq:chol}) gives
\begin{align}
P\tilde{\ve C}\left(\ve I-\frac{P}{1+KP}\ve E\right)\tilde{\ve C}^T=P\ve L\ve L^T
\label{eq:chol_sym}
\end{align}
where 
\begin{align}
\tilde{\ve C}=\begin{pmatrix}
\beta_1 &\beta_2 &\beta_3 &\ldots &\beta_{K}\\
0 &\beta_2 &0 &\ldots &0\\
0 &0 &\beta_3 &\ldots &0\\
\vdots &\vdots &\vdots &\ldots &\vdots\\
0  &0 &0 &0 &\beta_K
\end{pmatrix}
\end{align}
In this case the we are interested in the Cholesky factorization $\ve L$ above. Due to the special structure of $\ve A$ chosen in (\ref{eq:coef_sym}),   Theorem \ref{thm:K_MAC} implies that the following rates are achievable 
\begin{align}
R_1&= \frac{1}{2}\log\frac{\beta_1^2}{L_{11}^2}\\
R_k&= \min\left\{\frac{1}{2}\log\frac{\beta_k^2}{L_{11}^2},\frac{1}{2}\log\frac{\beta_k^2}{L_{kk}^2}\right\}, k\geq 2
\end{align}
Using the same argument in the proof of Theorem \ref{thm:K_MAC}, it is easy to show that the sum capacity is achievable if $L_{kk}^2\geq L_{11}^2$ for all $k\geq 2$. To achieve the symmetric capacity we further require that
\begin{align}
\frac{\beta_k^2}{L_{kk}^2}=\frac{\beta_j^2}{L_{jj}^2}
\label{eq:sym_eq_L}
\end{align}
for all $k,j$. This is the same as requiring $\ve B^{-1}\ve L$ to have  diagonals equal in amplitude with $\ve L$ given in (\ref{eq:chol_sym}), or equivalently requiring the matrix $\ve B^{-1}\ve A\ve B(\ve I+P\ve h\ve h^T)^{-1}\ve B^T\ve A^T\ve B^{-T}$ having Cholesky factorization whose diagonals are equal in amplitude. We can let $\beta_1=1$ without loss of generality and it is straightforward to check that in this case $\ve B^{-1}\ve A\ve B=\ve C$.  Now the condition in (\ref{eq:sym_eq_L}) is equivalently represented as 
\begin{align}
\tilde L_{kk}^2=\tilde L_{jj}^2
\end{align}
and the requirement $L_{kk}^2\geq L_{11}^2$ for $k\geq 2$ can be equivalently written as $\beta_k^2\geq \beta_1^2=1$.
\end{proof}

We point out that the value of power $P$ plays a key role in Proposition \ref{prop:C_sym}.  It is not true that for any power constraint $P$, there exists $\beta_2,\ldots,\beta_K$  such that the equality condition in Proposition \ref{prop:C_sym} can be fulfilled. For the two user case analyzed in Section \ref{sec:MAC_2usr}, we can show that for the symmetric channel, the equality condition in Proposition \ref{prop:C_sym} can be fulfilled  if the condition (\ref{eq:AchieCap_part}) holds, which in turn requires $P\geq 1.5$ for the symmetric channel. In general  for a given $K$, we expect that there exists a threshold $P^*(K)$ such that for $P\geq P^*(K)$, we can always find $\beta_2,\ldots,\beta_K$ which satisfy the equality condition in Proposition \ref{prop:C_sym} hence achieve the symmetric capacity. This conjecture is formulated as follows.

\begin{conjec}[Achievablity of the symmetric capacity]
For any $K\geq 2$, there exists a positive number $P^*(K)$, such that for all $P\geq P^*(K)$, we can find real numbers $\beta_2,\ldots,\beta_K$, where $|\beta_k|\geq 1$ with which the diagonal entries of  $\tilde{\ve L}$ given in (\ref{eq:L_tilde}) are equal in amplitude i.e., $|\tilde{L}_{kk}|=|\tilde{L}_{jj}|$ for all $k,j$.
\label{conj}
\end{conjec}

We have not been able to prove this claim. Table \ref{tab:beta_P15} gives some numerical results for the choices of $\underline{\beta}$ which achieve the symmetric capacity in a $K$-user Gaussian MAC with power constraint $P=15$ and different values of $K$. With this power constraint the claim in Conjecture \ref{conj} is numerically verified with $K$ up to $6$ in Table \ref{tab:beta_P15}. Notice that the value $\beta_k$ decreases with the index $k$ for $k\geq 2$. This is because with the  coefficient matrix $\ve A$ in (\ref{eq:coef_sym}), the decoding order of the individual users is from $2$ to $K$ (and user $1$ is decoded last). The earlier the message is decoded, the larger the corresponding $\beta$ will be.

\begin{table}[!hbt]
\renewcommand{\arraystretch}{1.3} \caption{The choice of $\underline{\beta}$ for a $K$-user Gaussian MAC with power $P=15$.}  \centering
\begin{tabular}{c||c|c|c|c|c|c} \hline
 $K$ &$\beta_1$ &  $\beta_2$ & $\beta_3$ & $\beta_4$ & $\beta_5$ &$\beta_6$\\ 
\hline\hline 2 &1 & 1.1438\\ 
\hline
3 &1	 &1.5853 &1.2582\\
\hline
4 &1 &1.6609 &1.3933 &1.1690\\
\hline
5	&1 &1.6909    &1.4626    &1.2796    &1.1034\\
\hline
6 &1 &1.6947 &1.4958 &1.3361 &1.1980 &1.0445
\end{tabular}
\label{tab:beta_P15}
 \end{table}
 
Some numerical results for $P^*(K)$ for $K$ up to $5$ is given in Table \ref{tab:Ps}. As we have seen $P^*(2)=1.5$. For other $K$ we give the interval which contains $P^*(K)$ by numerical evaluations. The interval containing  $P^*(K)$ for larger $K$ can be identified straightforwardly but the computation is time-consuming.
 \begin{table}[!hbt]
\renewcommand{\arraystretch}{1.3} \caption{The intervals containing $P^*(K)$}  \centering
\begin{tabular}{c||c} \hline
 $K$ &$P^*(K)$\\ 
\hline\hline 2 &1.5\\ 
\hline
3 & [2.23, 2.24]\\
\hline
4 &[3.74, 3.75]\\
\hline
5& [7.07, 7.08]
\end{tabular}
\label{tab:Ps}
 \end{table}

\section{The $2$-user Gaussian dirty MAC}\label{sec:dirty}
In the previous sections we  focused on the standard Gaussian multiple access channels. In this section we will consider the Gaussian MAC with interfering signals which are non-causally known at the transmitters. This channel model is called Gaussian ``dirty MAC" and is studied in \cite{philosof_lattice_2011}. Some related results are given in \cite{somekh-baruch_cooperative_2008}, \cite{kotagiri_multiaccess_2008}, 
 \cite{wang_approximate_2012}. A two-user Gaussian dirty MAC is given by
\begin{align}
\ve y=\ve x_1+\ve x_2+\ve s_1+\ve s_2+\ve z
\label{eq:doubly_dirty_mac}
\end{align}
where the channel input $\ve x_1, \ve x_2$ are required to satisfy the power constraints $\mathbb E\{\norm{\ve x_k}^2\}\leq P_k, k=1,2$ and $\ve z$ is the white Gaussian noise with unit variance per entry. The interference $\ve s_k$ is a zero-mean i.i.d. Gaussian random sequence with variance $Q_k$ for each entry, $k=1, 2$.  An important assumption is that the interference signal $\ve s_k$ is only  non-causally known to transmitter $k$. Two users need to mitigate two interference signals in a distributed manner, which makes this problem challenging. By letting $Q_1=Q_2=0$ we recover the standard Gaussian MAC. 

This problem can be seen as an extension of the well-known  dirty-paper coding problem \cite{costa_DPC_1983} to the multiple-access channels. However as shown in \cite{philosof_lattice_2011}, a straightforward extension of the usual Gelfand-Pinsker scheme \cite{Gelfand_Pinsker} is not optimal and in the limiting case when interference is very strong, the achievable rates are zero. Although the capacity region of this channel is unknown in general, it is shown in \cite{philosof_lattice_2011} that lattice codes are well-suited for this problem and give better performance than the usual ``random coding" scheme.

Now we will extend our coding scheme in previous sections to the dirty MAC. The basic idea is still to decode two linearly independent sums of the codewords. The new ingredient is to mitigate the interference $\ve s_1, \ve s_2$ in the context of lattice codes.  For a point-to-point AWGN channel with interference known non-causally at the transmitter, it has been shown that capacity can be attained with lattice codes \cite{zamir_nested_2002}.  Our coding scheme is an extension of the schemes in \cite{zamir_nested_2002} and \cite{philosof_lattice_2011}.

\begin{theorem}[Achievability for the Gaussian dirty MAC]
For the  dirty multiple access channel given in (\ref{eq:doubly_dirty_mac}), the following message rate pair is achievable
\begin{align*}
R_k=
\begin{cases}
r_k(\ve a,\underline{\gamma}, \underline{\beta}) &\quad\mbox{if }b_k=0\\
r_k(\ve b|\ve a,\underline{\gamma}, \underline{\beta}) &\quad\mbox{if }a_k=0\\
\min\{r_k(\ve a,\underline{\gamma}, \underline{\beta}), r_k(\ve b|\ve a,\underline{\gamma}, \underline{\beta})\}&\quad\mbox{otherwise}
\end{cases}
\end{align*}
for any linearly independent integer vectors $\ve a, \ve b\in\mathbb Z^{2}$ and $\underline{\gamma},\underline{\beta}\in\mathbb R^{2}$ if $r_k(\ve a,\underline{\gamma}, \underline{\beta})\geq 0$ and $r_k(\ve b|\ve a,\underline{\gamma}, \underline{\beta})>0$ for $k=1,2$, whose expressions are  given  as
\begin{align*}
r_k(\ve a,\underline{\gamma},\underline{\beta})&:=\max_{\alpha_1}\frac{1}{2}\log^+\frac{\beta_k^2 P_k}{N_1(\alpha_1,\underline{\gamma},\underline{\beta})}\\
r_k(\ve b|\ve a, \underline{\gamma},\underline{\beta})&:=\max_{\alpha_2,\lambda}\frac{1}{2}\log^+\frac{\beta_k^2 P_k}{N_2(\alpha_2,\underline{\gamma},\underline{\beta},\lambda)}
\end{align*}
with
\begin{align}
N_1(\alpha_1,\underline{\gamma},\underline{\beta})=&\alpha_1^2+\sum_{k=1}^2\Big((\alpha_1-a_k\beta_k)^2P_k\nonumber\\
&+(\alpha_1-a_k\gamma_k)^2Q_k\Big)
\label{eq:N_1}\\
N_2(\alpha_2,\underline{\gamma},\underline{\beta},\lambda)=&\alpha_2^2+\sum_{k=1}^2\bigg((\alpha_2-\lambda a_k\gamma_k-b_k\gamma_k)^2Q_k\nonumber\\
&+(\alpha_2-\lambda a_k\beta_k-b_k\beta_k)^2P_k\bigg)
\label{eq:N_2}
\end{align}

\label{thm:two_sums}
\end{theorem}

\begin{proof}
Let $\ve t_k$ be the lattice codeword of user $k$ and $\ve d_k$ the dither uniformly distributed in $\mathcal V_k^s/\beta_k$. The channel input is given as
\begin{align*}
\ve x_k=[\ve t_k/\beta_k+\ve d_k-\gamma_k \ve s_k/\beta_k]\mod\Lambda_k^s/\beta_k 
\end{align*}
for some $\gamma_k$ to be determined later. In Appendix \ref{sec:appen_dirty} we show that  with the channel output $\ve y$ we can form
\begin{IEEEeqnarray}{rCl}
\tilde {\ve y}_1&:=&\tilde {\ve z}_1+\sum_k a_{k} \tilde{\ve t}_k +\sum_k(\alpha_1-a_k\gamma_k)\ve s_k
\label{eq:y1_dirty}
\end{IEEEeqnarray}
where $\alpha_1$ is some real numbers to be optimized later and we define $\tilde{\ve t}_k :=\ve t_k-Q_{\Lambda_k^s}(\ve t_k+\beta_k\ve d_k-\gamma_k \ve s_k)$ and  $\tilde{\ve z}_1:=\sum_k(\alpha_1 -a_{k}\beta_k)\ve x_k+\alpha_1\ve z$. 
Due to the nested lattice construction we have $\tilde{\ve t}_k\in\Lambda$. Furthermore the term $\tilde{\ve z}_1 +\sum_k(\alpha_1-a_k\gamma_k)\ve s_k$ is independent of the sum $\sum_k a_k\tilde{\ve t}_k$ thanks to the dither and can be seen as the equivalent noise having  average power per dimension  $N_1(\alpha,\underline{\gamma},\underline{\beta})$ in (\ref{eq:N_1}) for $k=1, 2$.
In order to decode the integer sum $\sum_k a_k\tilde{\ve t}_k$ we require
\begin{align}
r_k<r_k(\ve a,\underline{\gamma},\underline{\beta}):=\max_{\alpha_1}\frac{1}{2}\log^+\frac{\beta_k^2 P_k}{N_1(\alpha_1,\underline{\gamma},\underline{\beta})}
\label{eq:rate_a}
\end{align}
Notice this constraint on $R_k$ is applicable only if $a_k\neq 0$.

If we can decode $\sum_ka_k\tilde{\ve t}_k$ with positive rate, the idea of successive interference cancellation can be applied.  We show in Appendix \ref{sec:appen_dirty} that for decoding the second sum we can form
\begin{IEEEeqnarray}{rCl}
\tilde{\ve y}_2&:=&\tilde{\ve z}_2+\sum_k(\alpha_2-\lambda a_k\gamma_k-b_k\gamma_k)\ve s_k+\sum_kb_k\tilde{\ve t}_k
\label{eq:y2_dirty}
\end{IEEEeqnarray}
where  $\alpha_2$ and $\lambda$ are two real numbers to be optimized later and we define $\tilde{\ve z}_2:=\sum_k(\alpha_2 -\lambda a_k\beta_k-b_k\beta_k)\ve x_k+\alpha_2\ve z$.  Now the equivalent noise $\tilde{\ve z}_2+\sum_k(\alpha_2-\lambda a_k\gamma_k-b_k\gamma_k)\ve s_k$  has average power per dimension $N_2(\alpha_2,\underline{\gamma},\underline{\beta},\lambda)$ given in (\ref{eq:N_2}). Using lattice decoding we can show the following rate pair for decoding $\sum_kb_k\tilde{\ve t}_k$ is achievable
\begin{IEEEeqnarray}{rCl}
r_k<r_k(\ve b|\ve a, \underline{\gamma},\underline{\beta}):=\max_{\alpha_2,\lambda}\frac{1}{2}\log^+\frac{\beta_k^2 P_k}{N_2(\alpha_2,\underline{\gamma},\underline{\beta},\lambda)}
\label{eq:rate_b}
\end{IEEEeqnarray}
Again the lattice points $\tilde{\ve t}_k$ can be solved from the two sums if $\ve a$ and $\ve b$ are linearly independent, and $\ve t_k$ is recovered by the modulo operation $\ve t_k=[\tilde{\ve t}_k]\mode\Lambda_k^s$ even if $\ve s_k$ is not known at the receiver.  If we have $b_k=0$, the above constraint does not apply to $R_k$.
\end{proof}

\subsection{Decoding one integer sum}\label{subsec:one_sum}
We revisit the results obtained in \cite{philosof_lattice_2011} and show they can be obtained in our framework in a unified way. 
\begin{theorem}[\cite{philosof_lattice_2011} Theorem 2, 3]
For the  dirty multiple access channel given in (\ref{eq:doubly_dirty_mac}), we have the following achievable rate region:
\begin{align*}
&R_1+R_2\\
&=\begin{cases}
\frac{1}{2}\log(1+\min\{P_1,P_2\})\: &\text{if }\sqrt{P_1P_2}-\min\{P_1,P_2\}\geq 1\\
\frac{1}{2}\log^+\left(\frac{P_1+P_2+1}{2+(\sqrt{P_1}-\sqrt{P_2})^2}\right)\:&\text{otherwise}
\end{cases}
\end{align*}
\label{thm:single_sum}
\end{theorem}

\begin{remark}The above rate region was obtained by considering the transmitting scheme where only one user transmits at a time. In our framework, it is the same as assuming one transmitted signal, say $\ve t_1$, is set to be $\ve 0$ and known to the decoder. In this case we need only one integer sum to decode $\ve t_2$. Here we give a proof  to show the achievability for
\begin{align}
R_2=
\begin{cases}
\frac{1}{2}\log(1+P_2)&\quad\text{for } P_1\geq \frac{(P_2+1)^2}{P_2}\\
\frac{1}{2}\log(1+P_1)&\quad\text{for } P_2\geq \frac{(P_1+1)^2}{P_1}\\
\frac{1}{2}\log^+\left(\frac{P_1+P_2+1}{2+(\sqrt{P_1}-\sqrt{P_2})^2}\right)&\quad\text{otherwise}
\end{cases}
\label{eq:single_sum_r2}
\end{align}
while $R_1=0$. Theorem \ref{thm:single_sum} is obtained by showing the same result holds when we switch the two users and a time-sharing argument.
\end{remark}

\begin{proof}
Choosing $\ve a=(1,1)$ and $\gamma_1=\gamma_2=\alpha_1$ in (\ref{eq:rate_a}), we can decode the integer sum $\sum_k \tilde{\ve t}_k$ if  
\begin{IEEEeqnarray}{rCl}
r_2<r_2(\ve a,\underline{\beta})
&=&\frac{1}{2}\log\frac{P_2(1+P_1+P_2)}{r^2P_1+P_2+P_1P_2(r-1)^2}
\label{eq:r2_r}
\end{IEEEeqnarray}
by choosing the optimal $\alpha_1^*=\frac{\beta_1P_1+\beta_2P_2}{P_1+P_2+1}$ and defining $r:=\beta_1/\beta_2$. An important observation is that in order to extract $\ve t_2$ from the integer sum (assuming $\ve t_1=\ve 0$)
\begin{align*}
\sum_k \tilde{\ve t}_k=\ve t_2-Q_{\Lambda_2^s}(\ve t_2+\beta_2\ve d_2-\gamma_2\ve s_2)-Q_{\Lambda_1^s}(\beta_1\ve d_1-\gamma_1\ve s_1), 
\end{align*}
one sufficient condition is $\Lambda_1^s\subseteq\Lambda_2^s$. Indeed, due to the fact that $[\ve x]\mode\Lambda_2^s=0$ for any $\ve x\in\Lambda_1^s\subseteq\Lambda_2^s$, we are  able to recover $\ve t_2$ by performing $[\sum_k\tilde{\ve t}_k]\mode\Lambda_2^s$ if $\Lambda_1^s\subseteq\Lambda_2^s$. This requirement amounts to the condition  $\beta_1^2P_1\geq \beta_2^2P_2$ or equivalently $r\geq \sqrt{P_2/P_1}$. Notice if we can extract $\ve t_2$ from just one sum $\sum_k\tilde{\ve t}_k$ (with $\ve t_1$ known), then the computation rate $R_2^{\ve a}=r_2(\ve a,\underline{\beta})$ will also be the message rate $R_2=r_2(\ve a,\underline{\beta})$.

Taking derivative w. r. t. $r$ in (\ref{eq:r2_r}) gives the critical point
\begin{align}
r^*=\frac{P_2}{P_2+1}
\end{align}
If $r^*\geq \sqrt{P_2/P_1}$ or equivalently $P_1\geq\frac{(P_2+1)^2}{P_2}$, substituting $r^*$ in (\ref{eq:r2_r}) gives
\begin{align*} 
R_2=\frac{1}{2}\log(1+P_2)
\end{align*}
If $r^*\leq  \sqrt{P_2/P_1}$ or equivalently $P_1\leq \frac{(P_2+1)^2}{P_2}$,  $R_2$ is non-increasing in $r$ hence we should choose $r=\sqrt{P_2/P_1}$ to obtain
\begin{align*}
R_2=\frac{1}{2}\log^+\left(\frac{1+P_1+P_2}{2+(\sqrt{P_2}-\sqrt{P_1})^2}\right)
\end{align*}
To show the result for the case $P_2\geq \frac{(P_1+1)^2}{P_1}$, we set the transmitting power of user $2$ to be $P_2'=\frac{(P_1+1)^2}{P_1}$ which is smaller or equal to its full power $P_2$ under this condition. In order to satisfy the nested lattice constraint $\Lambda_1^s\subseteq\Lambda_2^s$ we also need $\beta_1^2P_1\leq\beta_2^2P_2'$ or equivalently $r\geq \sqrt{P_2'/P_1}$. By replacing $P_2$ by the above $P_2'$ and choosing $r=\sqrt{P_2'/P_1}$ in (\ref{eq:r2_r}) we get
\begin{align}
R_2=\frac{1}{2}\log(1+P_1)
\end{align}
Interestingly under this scheme,  letting the transmitting power to be $P_2'$ gives a larger achievable rate than using the full power $P_2$  in this power regime.  
\end{proof}
An outer bound on the capacity region  given in \cite[Corollary 2]{philosof_lattice_2011} states that the sum rate capacity should satisfy 
\begin{align}
R_1+R_2\leq \frac{1}{2}\log(1+\min\{P_1,P_2\})
\end{align}
 for strong interference (both $Q_1, Q_2$ go to infinity). Hence in the strong interference case, the above achievability result is either optimal (when $P_1, P_2$ are not too close) or only a constant away from the capacity region (when $P_1, P_2$ are close, see \cite[Lemma 3]{philosof_lattice_2011}). However the rates in Theorem \ref{thm:single_sum} are strictly suboptimal for  general interference strength as we will show in the sequel.

\subsection{Decoding two integer sums}
Now we consider decoding two sums for the Gaussian dirty MAC by evaluating the achievable rates stated in Theorem \ref{thm:two_sums}.  Unlike  the case of the standard Gaussian MAC studied in Section \ref{sec:MAC_2usr}, here we need to optimize over $\underline{\gamma}$ for given $\ve a, \ve b$ and $\underline{\beta}$, which does not have a closed-form solution due to the $\min\{\cdot\}$ operation. Hence in this section we resort to numerical methods for evaluations. To give an example of the advantage for decoding two sums, we show achievable rate regions in Figure \ref{fig:RateRegion_doublyDirty} for a  dirty MAC where $P_1=Q_1=10$ and $P_2=Q_2=2$. We see in the case when the transmitting power and interference strength are comparable, decoding two sums gives a significantly larger achievable rate region. In this example we choose the coefficients to be $\ve a=(a_1,1),\ve b=(1,0)$  or $\ve a=(1,a_2),\ve b=(1,0)$  for $a_1,a_2=1,\ldots, 5$ and optimize over parameters $\underline{\gamma}$. We also point out that unlike the case of the clean MAC where it is best to choose $a_1, a_2$ to be $1$, here choosing coefficients $a_1, a_2$ other than $1$ gives larger achievable rate regions in general.

\begin{figure}[!h]
\centering
\includegraphics[scale=0.45]{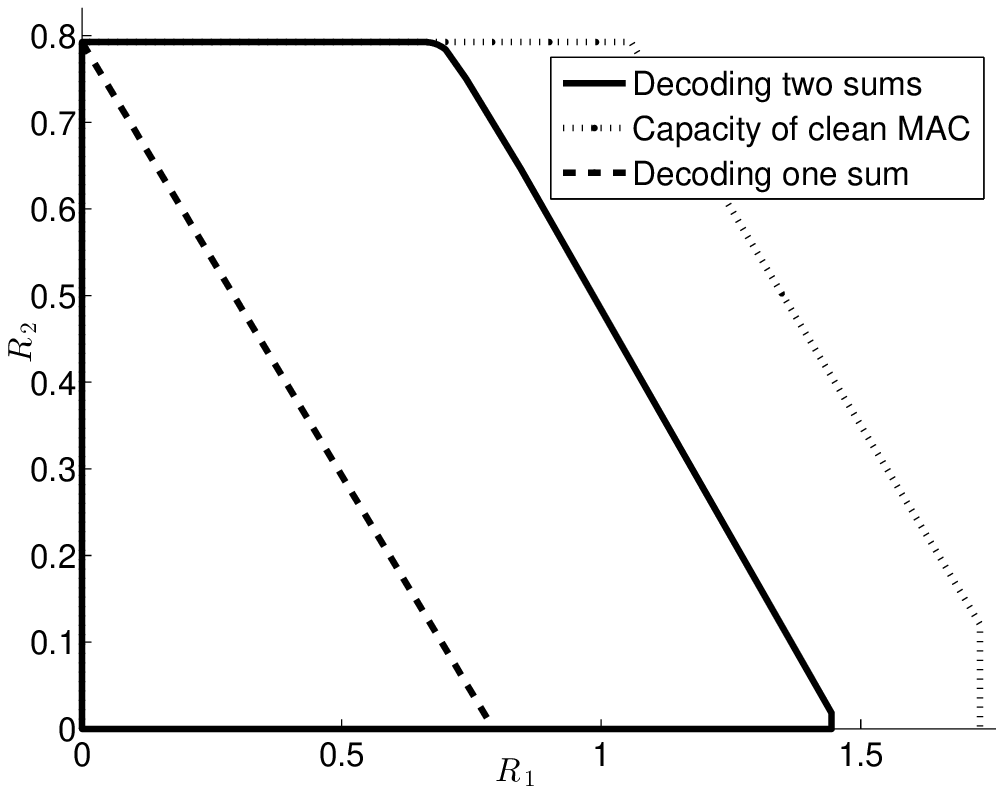}
\caption{We consider a  dirty MAC with $P_1=Q_1=10$ and $P_2=Q_2=2$. The dashed line is the achievable rate region given in Theorem \ref{thm:single_sum} from \cite{philosof_lattice_2011} which corresponds to decoding only one sum. The solid line gives the achievable rate region in Theorem \ref{thm:two_sums} by decoding two sums with the coefficients $\ve a=(a_1,1),\ve b=(1,0)$  or $\ve a=(1,a_2),\ve b=(1,0)$  for $a_1,a_2=1,\ldots, 5$ and optimizing over parameters $\underline{\gamma}$.}
\label{fig:RateRegion_doublyDirty}
\end{figure}

Different from the point-to-point Gaussian channel with interference known at the transmitter,  it is no longer possible to eliminate all interference completely without diminishing the capacity region for the  dirty MAC.  The proposed scheme provides  a way of trading off between eliminating the interference and treating it as noise. Figure \ref{fig:compare} shows the symmetric rate of the  dirty MAC as a function of interference strength. When the interference is weak, the proposed scheme balances the residual interference $\ve s_1, \ve s_2$ in $N_1$ and $N_2$  by optimizing the parameters $\underline{\gamma}$,  see Eqn. (\ref{eq:N_1}) and Eq. (\ref{eq:N_2}).  This is better than only decoding one sum in which we completely cancel out the interference.

\begin{figure}[!h]
\centering
\includegraphics[scale=0.55]{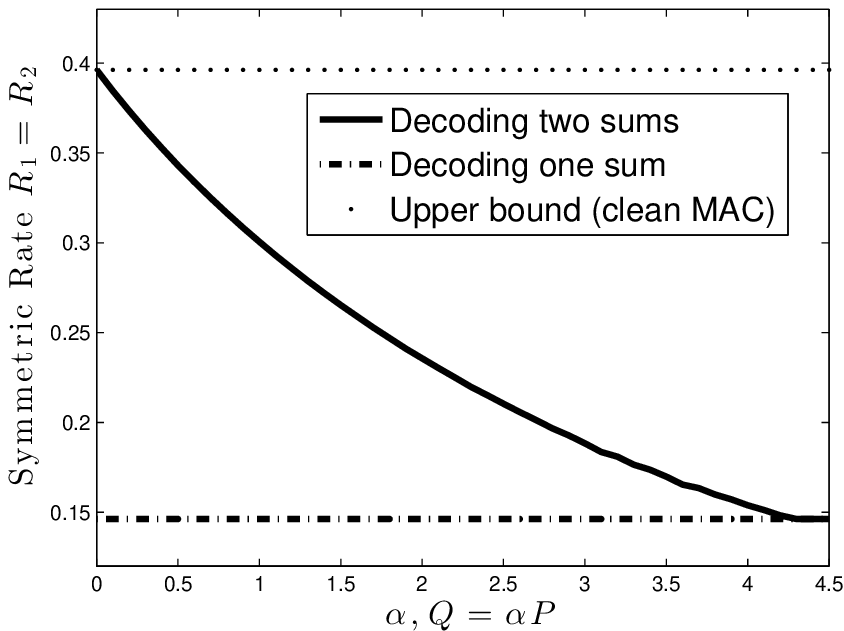}
\caption{We consider a  dirty MAC with $P_1=P_2=1$ and $Q_1=Q_2=\alpha P_1$ with different $\alpha$ varying from $[0,4.5]$. The vertical axis denotes the maximum symmetric rate $R_1=R_2$. The dotted line is the maximum symmetric rate $1/4\log(1+P_1+P_2)$ for a clean MAC as an upper bound. The dashed line gives the achievable symmetric rate in Theorem \ref{thm:single_sum} from \cite{philosof_lattice_2011} and the solid line depicts the symmetric rate in Theorem \ref{thm:two_sums} by decoding two sums. }
\label{fig:compare}
\end{figure}

As mentioned in the previous subsection, decoding one integer sum is near-optimal in the limiting case when both interference signals $\ve s_1, \ve s_2$ are very strong, i.e., $Q_1,Q_2\rightarrow \infty$. It is natural to ask if we can do even better by decoding two sums in this case.  It turns out in the limiting case we are not able to decode two linearly independent sums with this scheme. 
\begin{lemma}[Only one sum for high interference]
For the $2$-user  dirty MAC  in (\ref{eq:doubly_dirty_mac}) with $Q_1, Q_2\rightarrow\infty$, we have $r_k(\ve a,\underline{\gamma},\underline{\beta})=r_k(\ve b|\ve a,\underline{\gamma},\underline{\beta})=0, k=1,2 $ for any linearly independent $\ve a, \ve b$ where $a_k\neq 0$, $k=1,2$.
\end{lemma}
\begin{proof}
The  rate expressions in (\ref{eq:rate_a}) and (\ref{eq:rate_b}) show that we need to eliminate all terms involving $Q_k$ in the equivalent noise $N_1$ in (\ref{eq:N_1}) and $N_2$ in (\ref{eq:N_2}),  in order to have a positive rate when $Q_1,Q_2\rightarrow\infty$. Consequently we need $\alpha_1-a_k\gamma_k=0$ and $\alpha_2-\lambda a_k\gamma_k-b_k\gamma_k=0$ for $k=1,2$.  
or equivalently
\begin{align}
\begin{pmatrix}
1 &0 &-a_1 &0\\
0 &1 &-\lambda a_1-b_1 &0\\
1 &0 &0 &-a_2\\
0 &1 &0 &-\lambda a_2-b_2
\end{pmatrix}
\begin{pmatrix}
\alpha_1\\
\alpha_2\\
\gamma_1\\
\gamma_2
\end{pmatrix}
=\ve 0
\end{align}
Performing elementary row operations gives the following equivalent system
\begin{align*}
\begin{pmatrix}
1 &0 &-a_1 &0\\
0 &1 &-\lambda a_1-b_1 &0\\
0 &0 &a_1 &-a_2\\
0 &0 &0 &\frac{a_2(\lambda a_1+b_1)}{a_1}-\lambda a_2-b_2
\end{pmatrix}
\begin{pmatrix}
\alpha_1\\
\alpha_2\\
\gamma_1\\
\gamma_2
\end{pmatrix}
=\ve 0
\end{align*}
To have non-trivial solutions of $\underline{\alpha}$ and $\gamma$ with $a_1\neq 0$, we must have $\frac{a_2(\lambda a_1+b_1)}{a_1}-\lambda a_2-b_2=0$, which simplifies to $a_2b_1=a_1b_2$, meaning $\ve a$ and $\ve b$ are linearly dependent.
\end{proof}

This observation suggests that when both interference signals are very strong,  the strategy in \cite{philosof_lattice_2011} to let only one user transmit at a time (section \ref{subsec:one_sum}) is the best thing to do within this framework.  However we point out that in the case when only one interference is very strong, we can still decode two independent sums with positive rates.  For example consider the system in (\ref{eq:doubly_dirty_mac}) with $\ve s_2$ being identically zero, $\ve s_1$  only known to User 1  and $Q_1\rightarrow\infty$.  In this case we can decode two linearly independent sums with  $\ve a=(1,1), \ve b=(1,0)$ or  $\ve a=(1,0), \ve b=(0,1)$. The resulting achievable rates with Theorem \ref{thm:two_sums}  is the same as that given in \cite[Lemma 9]{philosof_lattice_2011}.  Moreover, the capacity region of the dirty MAC with only one interference signal commonly known to both users \cite[VIII]{philosof_lattice_2011} can also be achieved using Theorem \ref{thm:two_sums}, by choosing $\ve a=(1,0), \ve b=(0,1)$ for example.

\section{Concluding Remarks}\label{sec:conclusion}
We have shown that the CFMA strategy is able to achieve the capacity of Gaussian multiple access channels. This coding scheme is of possible practical interests because it only uses a single-user decoder without time-sharing or rate-splitting techniques. The proposed coding scheme is a generalization of the  compute-and-forward  technique and can be applied to many other Gaussian network scenarios. However as we have seen in the $2$-user and $3$-user examples, this scheme fails to have its capacity-achieving ability if the signal-to-noise ratio is very low. The reason  lies in the fact that in this regime the computation rate pair is not large enough (recall the examples in Theorem \ref{thm:clean_MAC}).  Hence  it is interesting to ask  what is the largest possible computation rate tuple in a Gaussian MAC.  The answer to this question has implications on many other unsolved network communication problems, including the Gaussian interference channel and the Gaussian two-way relay channel.

\appendices

\section{The proof of Theorem \ref{thm:computation_rate}}\label{sec:appen_proof_computation}


A  proof of the $2$-user case of Theorem \ref{thm:computation_rate} is already contained in the proof of Theorem \ref{thm:MAC_general}. We now give a  proof for the $K$-user case.

In this case the receiver only decodes one sum, and we can choose all the fine lattices $\Lambda_k, k=1,\ldots K$ to be the same lattice, denoted as $\Lambda$. When the message (codeword) $\ve t_k$ is given to encoder $k$, it forms its channel input as follows
\begin{align}
\ve x_k=\left[\ve t_k/\beta_k+\ve d_k\right]\mode\Lambda_k^s/\beta_k
\end{align}
where the dither $\ve d_k$  is a random vector uniformly distributed in the scaled Voronoi region $\mathcal V_k^s/\beta_k$.  Notice that $\ve x_k$ is independent from $\ve t_k$ and also uniformly in $\Lambda_k^s/\beta_k$ hence  has average power $P$ for all $k$.

At the decoder we form
\begin{align*}
\tilde {\ve y}:=&\alpha\ve y-\sum_ka_{k}\beta_k\ve d_k\\
=&\sum_k a_{k}\left(\beta_k (\ve t_k/\beta_k+\ve d_k)-\beta_k Q_{\Lambda_k^s/\beta_k}(\ve t_k/\beta_k+\ve d_k)\right)\\
&-\sum_ka_{k}\beta_k\ve d_k+\tilde{\ve  z}\\
\stackrel{(a)}{=}&\tilde {\ve z}+\sum_k a_{k} (\ve t_k-Q_{\Lambda_k^s}(\ve t_k+\beta_k\ve d_k))\\
:=&\tilde {\ve z}+\sum_k a_{k} \tilde{\ve t}_k
\end{align*}
with $\tilde {\ve t}_k:=\ve t_k-Q_{\Lambda_k^s}(\ve t_k+\beta_k\ve d_k)$ and the equivalent noise 
\begin{IEEEeqnarray}{rCl}
\tilde {\ve z}:=\sum_k(\alpha h_{k}-a_{k}\beta_k)\ve x_k+\alpha\ve z
\end{IEEEeqnarray}
which is independent of $\sum_k a_{k}\tilde{\ve t}_k$ since all $\ve x_k$ are independent of $\sum_k a_{k}\tilde {\ve t}_k$ thanks to the dithers $\ve d_k$. The step $(a)$ follows because it holds $Q_{\Lambda}(\beta  X)=\beta Q_{\frac{\Lambda}{\beta}}(X)$ for any $\beta\neq 0$. 
Notice we have $\tilde {\ve t}_k\in\Lambda$ since  $\ve t_k\in\Lambda$ and $\Lambda_k^s\subseteq\Lambda$ due to the code construction.  Hence the linear combination $\sum_k a_{k}\tilde{\ve t}_k$ along belongs to the decoding lattice $\Lambda$.

The decoder uses lattice decoding to obtain $\sum_ka_{k}\tilde{\ve t}_k$ with respect to the decoding lattice $\Lambda$ by quantizing $\tilde {\ve y}$ to its nearest neighbor in  $\Lambda$. The decoding error probability  is equal to the probability that the equivalent noise $\tilde {\ve z}$ leaves the Voronoi region surrounding the lattice point representing $\sum_ka_{k}\tilde{\ve t}_k$. If the fine lattice $\Lambda$ is good for AWGN channel, as it is shown in \cite{ErezZamir_2004}, the probability $\prob{\tilde {\ve z}\notin \mathcal V}$ goes to zero exponentially if 
\begin{IEEEeqnarray}{rCl}
\frac{\vol(\mathcal V)^{2/n}}{N(\alpha)}> 2\pi e
\label{eq:volume_to_noise}
\end{IEEEeqnarray}
where 
\begin{align}
N(\alpha):=\mathbb E\norm{\tilde {\ve z}}^2/n=\norm{\alpha\ve h-\tilde{\ve a}}^2P+\alpha^2
\end{align}
denotes the average power per dimension of the equivalent noise. Recall that the shaping lattice $\Lambda_k^s$ is good for quantization hence we have 
\begin{IEEEeqnarray}{rCl}
\vol(\mathcal V_k^s)=\left(\frac{\beta_k^2P}{G(\Lambda_k^s)}\right)^{n/2}
\end{IEEEeqnarray}
with $G(\Lambda_k^s)2\pi e<(1+\delta)$ for any $\delta>0$ if $n$ is large enough \cite{ErezZamir_2004}. Together with the message rate expression in (\ref{eq:message_rate_lattice}) we can see that lattice decoding is successful if $\beta_k^2P2^{-2R_k}/G(\Lambda_k^s)>2\pi e N$ for every $k$, or equivalently
\begin{IEEEeqnarray*}{rCl}
r_k<\frac{1}{2}\log\left(\frac{P}{N(\alpha)}\right)+\frac{1}{2}\log\beta_k^2-\frac{1}{2}\log(1+\delta)
\end{IEEEeqnarray*}
By choosing $\delta$ arbitrarily small and optimizing over $\alpha$ we conclude that the lattice decoding of $\sum_ka_k\tilde{\ve t}_k$ will be successful if
\begin{align}
r_k<\max_{\alpha}\frac{1}{2}\log\left(\frac{P}{N(\alpha)}\right)+\frac{1}{2}\log\beta_k^2=R_k^{\ve a}
\end{align}
with $R_k^{\ve a}$ given in (\ref{eq:compute_rate}). Lastly the modulo sum is obtained by
\begin{align*}
&\left[\sum_ka_k\tilde t_k\right]\mod\Lambda_f^s\\
&=\left[\sum_ka_k\ve t_k-\sum_ka_kQ_{\Lambda_k^s}(\ve t_k+\beta_k\ve d_k)\right]\mod\Lambda_f^s\\
&=\left[\sum_ka_kt_k\right]\mod\Lambda_f^s
\end{align*}
where the last equality holds because $\Lambda_f^s$ is the finest lattice among $\Lambda_k^s, k=1,\ldots, K$.

\section{Derivations in the proof of Theorem \ref{thm:clean_MAC}}\label{sec:appen_proof_clean}
Here we prove the claim in Theorem \ref{thm:clean_MAC} that $\beta_2^{(1)},\beta_2^{(2)}\in[\beta_2',\beta_2'']$ if and only if the Condition (\ref{eq:AchieCap_whole}) holds. Recall we have defined $\beta_2^{(1)}:=\frac{h_1h_2P}{1+h_1^2P}$, $\beta_2^{(2)}:=\frac{1+h_2^2P}{h_1h_2P}$ and $\beta_2',\beta_2''$ in Eqn. (\ref{eq:roots}).

With the choice $\ve a=(1,1)$ we can rewrite (\ref{eq:roots}) as
\begin{align}
\beta_2':=\frac{2h_1h_2P+S-\sqrt{SD}}{2(1+h_1^2P)}\\
\beta_2'':=\frac{2h_1h_2P+S+\sqrt{SD}}{2(1+h_1^2P)}
\end{align}
with $S:=\sqrt{1+h_1^2P+h_2^2P}$ and $D:=4Ph_1h_2-3S$. Clearly the inequality $\beta_2'\leq \beta_2^{(1)}$ holds if and only if $S-\sqrt{SD}\leq 0$ or equivalently
\begin{align}
\frac{Ph_1h_2}{\sqrt{1+h_1^2P+h_2^2P}}\geq 1
\end{align}
which is just  Condition (\ref{eq:AchieCap_whole}). Furthermore notice that $\beta_2^{(1)}<\frac{h_2}{h_1}P<\beta_2^{(2)}$ hence it remains to prove that $\beta_2^{(2)}\leq \beta_2''$ if and only if  (\ref{eq:AchieCap_whole}) holds. But this follows immediately by noticing that $\beta_2^{(2)}\leq \beta_2''$ can be rewritten as
\begin{align}
2S^2\leq h_1h_2P(S+\sqrt{SD})
\end{align}
which is satisfied if and only if $S\leq D$, or equivalently  Condition (\ref{eq:AchieCap_whole}) holds.

\section{The derivations in the proof of Theorem \ref{thm:two_sums}}\label{sec:appen_dirty}
In this section we give the derivation of the expressions of $\tilde{\ve y}_1$ in (\ref{eq:y1_dirty}) and $\tilde{\ve y}_2$ in  (\ref{eq:y2_dirty}). To obtain $\tilde{\ve y}_1$, we process the channel output $\ve y$ as
\begin{IEEEeqnarray*}{rCl}
\tilde {\ve y}_1&:=&\alpha_1\ve y-\sum_ka_{k}\beta_k\ve d_k\nonumber\\
&=&\underbrace{\sum_k(\alpha_1 -a_{k}\beta_k)\ve x_k+\alpha_1\ve z}_{\tilde{\ve z}_1}+\alpha_1 \sum_k\ve s_k\\
&&+\sum_k a_{k}\beta_k\ve x_k-\sum_ka_{k}\beta_k\ve d_k\\
&=&\tilde{\ve  z}_1+\alpha_1\sum_k\ve s_k+\sum_k a_{k}\beta_k (\ve t_k/\beta_k+\ve d_k-\gamma_k\ve s_k/\beta_k)\\
&&-\sum_k a_k\beta_k Q_{\Lambda_k^s/\beta_k}(\ve t_k/\beta_k+\ve d_k-\gamma_k\ve s_k/\beta_k)-\sum_ka_{k}\beta_k\ve d_k\\
&=&\tilde {\ve z}_1+\sum_k a_{k} (\ve t_k-Q_{\Lambda_k^s}(\ve t_k+\beta_k\ve d_k-\alpha_1 \ve s_k))\\
&&+\sum_k(\alpha_1-a_k\gamma_k)\ve s_k\\
&=&\tilde {\ve z}_1+\sum_k a_{k} \tilde{\ve t}_k +\sum_k(\alpha_1-a_k\gamma_k)\ve s_k
\end{IEEEeqnarray*}

When the sum $\sum_k a_{k} \tilde{\ve t}_k$ is  decoded,  the term $\tilde{\ve z}_1 +\sum_k(\alpha_1-a_k\gamma_k)\ve s_k$ which can be calculated using $\tilde{\ve y}_1$ and $\sum_ka_k\tilde{\ve t}_k$.  For decoding the second sum we form the following with some numbers $\alpha_2'$ and $\lambda$:
\begin{IEEEeqnarray*}{rCl}
\tilde{\ve y}_2&:=&\alpha_2'\ve y+\lambda\left(\tilde{\ve z}_1 +\sum_k(\alpha_1-a_k\gamma_k)\ve s_k\right)-\sum_kb_k\beta_k\ve d_k\\
&=&\alpha_2'(h_1\ve x_1+h_2\ve x_2+\ve s_1+\ve s_2+\ve z)\\
&&+\sum_k(\lambda\alpha_1h_k-\lambda a_k\beta_k)\ve x_k\\
&&+\lambda\alpha_1\ve z+\lambda\sum_k(\alpha_1-a_k\gamma_k)\ve s_k\\
&=&\sum_k(\alpha_2' +\lambda \alpha_1-\lambda a_k\beta_k)\ve x_k+(\alpha_2'+\lambda\alpha_1)\ve z\\
&&+\sum_k(\alpha_2'+\lambda\alpha_1- \lambda a_k\gamma_k)\ve s_k-\sum_kb_k\beta_k\ve d_k\\
&: =&\sum_k(\alpha_2-\lambda a_k\beta_k)\ve x_k+\alpha_2\ve z\\
&&+\sum_k(\alpha_2- \lambda a_k\gamma_k)\ve s_k-b_k\beta_k\ve d_k
\end{IEEEeqnarray*}
by defining $\alpha_2:=\alpha_2'+\lambda\alpha_1$. In the same way as deriving  $\tilde{\ve y}_1$, we can show
\begin{IEEEeqnarray*}{rCl}
\tilde{\ve y}_2&=&\underbrace{\sum_k(\alpha_2 -\lambda a_k\beta_k-b_k\beta_k)\ve x_k+\alpha_2\ve z}_{\tilde{\ve z}_2}\\
&&+\sum_k(\alpha_2-\lambda a_k\gamma_k)\ve s_k+\sum_k b_k\beta_k\ve x_k-\sum_kb_k\beta_k\ve d_k\\
&=&\tilde{\ve z}_2+\sum_k(\alpha_2-a_k\gamma_k)\ve s_k+ \\
&&\sum_k b_{k}\bigg(\beta_k (\ve t_k/\beta_k+\ve d_k-\gamma_k\ve s_k/\beta_k)\\
&&-\beta_k Q_{\Lambda_k^s/\beta_k}(\ve t_k/\beta_k+\ve d_k-\gamma_k\ve s_k/\beta_k)\bigg)-\sum_kb_k\beta_k\ve d_k\\
&=&\tilde{\ve z}_2+\sum_k(\alpha_2-\lambda a_k\gamma_k-b_k\gamma_k)\ve s_k+\sum_kb_k\tilde{\ve t}_k
\end{IEEEeqnarray*}
by defining $\alpha_2:=\alpha_2'+\lambda\alpha_1$ and $\tilde{\ve z}_2:=\sum_k(\alpha_2 -\lambda a_k\beta_k-b_k\beta_k)\ve x_k+\alpha_2\ve z$.

\section*{Acknowledgment}
The authors wish to thank Sung Hoon Lim,  Bobak Nazer,  Chien-Yi Wang and anonymous reviewers for helpful comments.

\bibliographystyle{IEEEtran}
\bibliography{IEEEabrv,MAC_CF_full}

\begin{IEEEbiographynophoto}{Jingge Zhu}
received the B.S. degree and M.S. degree in electrical engineering from Shanghai Jiao Tong University, Shanghai, China, in 2008 and 2011, respectively,  the Dipl.-Ing. degree in technische Informatik from Technische Universit{\"a}t Berlin, Berlin, Germany in 2011 and the
Doctorat \`es Science degree from the Ecole Polytechnique F\'ed\'erale (EPFL), Lausanne, Switzerland, in 2016. His research interests include  information theory with applications in communication systems.

Mr. Zhu is the recipient of the IEEE Heinrich Hertz Award for Best
Communications Letters in 2013. He also received the Early Postdoc.Mobility Fellowship from Swiss National Science Foundation in 2015.
\end{IEEEbiographynophoto}

\begin{IEEEbiographynophoto}{Michael Gastpar}
received the Dipl. El.-Ing. degree from the Eidgen\"ossische Technische Hochschule (ETH), Z\"urich, Switzerland, in 1997, the M.S. degree in electrical engineering from the University of Illinois at Urbana-Champaign, Urbana, IL, USA, in 1999, and the
Doctorat \`es Science degree from the Ecole Polytechnique F\'ed\'erale (EPFL), Lausanne, Switzerland, in 2002. He was also a student in engineering and philosophy at the Universities of Edinburgh and Lausanne.

During the years 2003-2011, he was an Assistant and tenured Associate Professor in the Department of Electrical Engineering and Computer Sciences at the University of California, Berkeley. Since 2011, he has been a Professor in the School of Computer and Communication
Sciences, Ecole Polytechnique F\'ed\'erale (EPFL), Lausanne, Switzerland.
He was also a professor at Delft University of Technology, The Netherlands, and
a researcher with the Mathematics of Communications Department,
Bell Labs, Lucent Technologies, Murray Hill, NJ.
His research interests are
in network information theory and related coding and signal processing techniques,
with applications to sensor networks and neuroscience.

Dr. Gastpar received the IEEE Communications Society and Information Theory Society Joint Paper Award in 2013 and the EPFL Best Thesis Award in 2002. He was an Information Theory Society Distinguished Lecturer (2009-2011), an Associate Editor for Shannon Theory for the IEEE TRANSACTIONS ON INFORMATION THEORY (2008-2011), and he has served as Technical Program Committee Co-Chair for the 2010 International Symposium on Information Theory, Austin, TX.
\end{IEEEbiographynophoto}

\end{document}